%% file: main_arxiv.tex
\documentclass[12pt,a4paper]{article}
\usepackage{pst-all}
\usepackage{inputenc}
\usepackage{xcolor}
\usepackage{bm}
\usepackage{breqn}
\usepackage{mathtools}
\usepackage{tabularx}
\usepackage{booktabs}
\usepackage{wrapfig}
\usepackage{lscape}
\usepackage{rotating}
\usepackage{multirow}
\usepackage{multicol}
\usepackage{amsmath}
\usepackage{amsthm}
\usepackage{amssymb} 
\usepackage{placeins}
\usepackage{authblk}
\usepackage{natbib}
\usepackage{graphicx}
\graphicspath{{Figures/}}
\usepackage{subfig}
\usepackage{url}
\usepackage{mathrsfs}
\usepackage{dsfont}
\usepackage{geometry}
\usepackage{bbm}
\usepackage{bm}
\usepackage{hyperref}
\usepackage{cleveref}
\usepackage{algorithm}
\usepackage{algpseudocode}
\usepackage{enumitem}
\usepackage{caption}

\usepackage{tikz}
\usepackage{standalone}
\usetikzlibrary{calc}
\usepackage{adjustbox}

\usepackage{appendix}
\usepackage{natbib}

\providecommand{\keywords}[1]
{
  \small	
  \textbf{\textit{Keywords---}} #1
}

\providecommand{\JEL}[1]
{
  \small	
  \textbf{\textit{JEL Codes:}} #1
}

\DeclareMathOperator*{\plim}{plim}

\graphicspath{{Thesis/Figures/}}

\theoremstyle{definition}
\newtheorem{theorem}{Theorem}

\newtheorem{definition}{Definition}
\newtheorem{proposition}{Proposition}
\newtheorem{assumption}{Assumption}

\usepackage{chngcntr}
\let\origtheassumption\theassumption

\usepackage[T1]{fontenc}
\usepackage{rotating}

\usepackage{scrwfile}
\TOCclone[\contentsname]{toc}{atoc}

\AfterTOCHead[toc]{%
}
\AfterTOCHead[atoc]{%
  \edef\maintocdepth{\the\value{tocdepth}}%
  \value{tocdepth}=-10000\relax%
}

\title{Estimating Spillover Effects from Sampled Connections\thanks{We thank Stanislav Avdeev, Vasco Carvalho, Jonathan Dingel, Alan Griffith, Eyo Herstad, Chih-Sheng Hseih, Max Kasy, Fran\c{c}ois Lafond, Xiaodong Liu, Jos van Ommeren, Xun Tang, Sander de Vries, Lina Zhang, and participants at the Network Science in Economics Conference 2025, European Summer Meeting of the Econometric Society 2024, University of Warwick, Vrije Universiteit Amsterdam and Tinbergen Institute for comments. The author also thanks the Smith School of Enterprise and the Environment at the University of Oxford for hospitality while preparing initial parts of this draft. The usual disclaimer applies.}}
\author[1]{Kieran Marray}
\affil[1]{Vrije Universiteit Amsterdam and Tinbergen Institute}

\date{September 2025\\
\bigskip
\href{https://kieranmarray.com/images/Estimating_spillover_effects.pdf}{\small (Link to most recent version)} }

\begin{document}

\maketitle

\begin{abstract}
Empirical researchers often estimate spillover effects by fitting linear or non-linear regression models to sampled network data. We show that common sampling schemes bias these estimates, potentially upwards, and derive biased-corrected estimators that researchers can construct from aggregate network statistics. Our results apply under different assumptions on the relationship between observed and unobserved links, allow researchers to bound true effect sizes, and to determine robustness to mismeasured links. As an application, we estimate the propagation of climate shocks between U.S. public firms from self-reported supply links, building a new dataset of county-level incidence of large climate shocks.
\end{abstract}


\JEL{C21}

\section{Introduction}\label{sec:Intro}
Empirical researchers measuring spillover effects often observe networks imperfectly, sampling either too few or too many links between individuals \citep{Newman2010}. In economics of education and development economics, for instance, researchers often collect network data through surveys that ask subjects to name up to a certain number of friends or contacts \citep[e.g][]{ Rapoport1961, Harris2009, Calvó-Armengol2009, Conley2010, Oster2012, Banerjee2013, Shakya2017}. In industrial organisation and economics of innovation, technological similarity or geographic proximity are often used to proxy firm connections \citep[e.g][]{Jaffe1986, Foster1995, Bloom2013}. When studying production networks, researchers often observe only large supply relationships between firms \citep[e.g see][]{Atalay2011, Barrot2016} or payments recorded by a specific bank or credit rating agency \citep[e.g][]{Carvalho2020}.\footnote{Other examples of researchers using proxies for links between individuals include in estimates of neighbourhood spillovers in crime \citep{Glaeser1996}, the role of social networks in labour markets \citep{Munshi2003, Beaman2011}, and the effect of deworming on educational outcomes \citep{Miguel2004}.} To illustrate the prevalence of this, we surveyed articles published in the American Economic Review, Econometrica, or Quarterly Journal of Economics from January 2020-September 2024. Out of the 30 papers measuring spillovers, 21 ($70\%$) sample the network imperfectly. 

A common empirical strategy is to estimate the spillover effect from some treatment by regressing outcomes on the (weighted) sum of treatments of sampled neighbours. We show that sampling links can bias these estimates even when treatment is randomly assigned in controlled or natural experiments. Sampling links can create an omitted variable -- the (weighted) sum of treatments of unobserved neighbours -- that affects outcomes and covaries with the regressor through the sampling scheme. Unlike attenuation bias from classical measurement error, estimates can be biased upward or downward. Simulations suggest that biases can be economically significant. For example, applying the sampling rule for females friends from the popular National Longitudinal Adolescent Health Data Set \citep{Harris2009} to simulated networks leads to estimates that are over one and a half times true spillover effects on average. A sufficient condition for sampling to cause bias is that all links have the same sign, and that the researcher systematically samples too few or too many links -- typical in social and economic network datasets  \citep{Harris2009, Banerjee2013, Barrot2016}. 

 %We first show that common sampling schemes induce dependence between observed and unobserved spillovers, even when treatment is independently and identically distributed across individuals. Dependence between observed and unobserved spillovers biases regression estimates of spillover effects upwards when the dependence is positive, and downwards when the dependence is negative. Estimates of the average total spillover effect for individuals with at least one treated neighbour are biased downwards. The size of biases can be economically significant. For example, applying the sampling rule from the popular National Longitudinal Adolescent Health Data Set \citep{Harris2009} to simulated networks leads to ordinary least-squares estimates that are over one and a half times true spillover effects on average. %Mean estimated spillover effects from nonlinear models are nearly double true spillover effects.

As collecting comprehensive network data is very costly \citep{Newman2010, Beaman2021}, we derive bias-corrected estimators for spillover effects in linear and non-linear social network models using sampled network data. Researchers must rescale estimates to adjust for the dependence between spillovers they observe and the spillovers they do not observe induced by their sampling rule. When treatment is distributed independently of network structure, as in a randomised or natural experiment, researchers can implement our correction using average numbers of missing links —- quantities that can be collected with a single additional survey question (e.g., “How many friends do you have?”) or obtained from external datasets that survey comparable networks more comprehensively (e.g see \citet{Jackson2022} for study partnerships at universities, \citet{Bacilieri2023} for firm-level supply relationships). When treatment assignment depends upon the network structure, we show how researchers can still correct estimates by modeling the dependence between treatment assignment and network structure using a copula.

 Our bias-corrected estimators are consistent, asymptotically normal. and perform well in simulation under common sampling rules where uncorrected estimators are severely biased. We further show how researchers can construct standard errors accounting for the uncertainty in the necessary network statistics using a bootstrap. If researchers do not observe the necessary network statistics, we show how they can bound spillover effects and assess robustness of estimates to measurement error instead.

 As an application, we estimate how large climate shocks propagate between U.S. public firms using a popular dataset containing self-reported supply links \citep{Atalay2011}. As firms are only mandated to report customers making up more than $10 \%$ of sales, the dataset under-samples their supply relationships. We combine the dataset with a newly constructed county-level measure of exposure to large weather shocks, and then estimate spillover effects correcting for sampling bias using network statistics from \cite{Bacilieri2023, Herskovic2020}. Corrected estimates are half the size of standard regression estimates. In the appendix, we also show that undersampling study partnerships between high and low-ability students can help account for differences between estimated and realised peer effects in \cite{Carrell2013}.
 
Our paper relates to a large literature in measurement error in econometrics in general \citep[e.g see][]{Heckman1979, Bound2001}, and contributes to the nascent literature on the effect of misspecification in network econometrics in particular \citep{Chandrasekhar2016,  Griffith2022, Lewbel2023, Yauck2022, Zhang2023, Hseih2024, Griffith2024, Boucher2025}. We differ from these other papers by focusing on settings where treatment are distributed independently from link strength (such as randomised or natural experiments) -- increasingly common in applied research \citep{Borusyak2024}. The setting allows us to construct bias-corrected estimators without imposing parametric assumptions about the network generation process \citep[as in, for example,][]{Breza2020, Boucher2025, Herstad2023}, or assuming constant link-missingness rates \citep[as in][]{Lewbel2025}, as typical in the existing literature. This matters because researchers might be hesitant to condition estimates on strong assumptions about unobservables. By contrast, our bias-corrected estimators depend only upon quantities that can be directly sampled. \cite{Chandrasekhar2016} suggests that researchers simply drop observations that might be incorrectly sampled in regression estimators. Our results do not require individuals to drop observations, which is especially useful when researchers do not know which proxies for links might be accurate or not. Our results nest those in \cite{Griffith2022} for the specific case of fixed choice designs, which he analyses in detail. The idea of using additional network data is similar to \cite{Lewbel2023, Zhang2023}. But we do not require researchers to collect an entire different measure of the same network in detail. Our results are also closely related to the literature on design based estimation using linear combinations of exposures to exogenous shocks \citep{Borusyak2023, Borusyak2024}. Again, our approach differs by not requiring researchers to specify a counterfactual distribution of exposure to exogenous shocks to correct for bias in regression estimates.

A related literature on unobserved networks assumes that researchers observe no links at all, and seeks to estimate the missing links between individuals \citep[e.g see][]{Manresa2013, Lam2019, Battaglini2021, Higgins2023, Lewbel2023, Rose2023, DePaula2024, Griffith2024, Marray2025}. A researcher could adopt a similar approach to spillover estimation on a mismeasured network: first estimate the true network from the mismeasured one, then use the estimated network to recover spillover parameters. This strategy requires richer data -— typically a short panel of individual outcomes -— and stronger structural assumptions on the data-generating process than our approach \citep[e.g.][]{Battaglini2021, DePaula2024}. Moreover, measurement error in the estimated network may itself bias regression estimates of spillover effects.

We proceed as follows. In Section \ref{sec:linear_models}, we characterise the effect of sampling links on linear regression estimates of spillover effects, and present bias-corrected estimators. Section \ref{sec:nonlinear_models} extends our results to common non-linear models, and \ref{sec:extensions} to cases when treatment depends on network structure. In Section \ref{sec:simulations}, we assess performance estimators by simulation. Finally, Section \ref{sec:empirical_application} presents our empirical examples. Proofs and additional results are provided in the appendix.

 %%%%%%%%%%%%%%%%%%%%%%%%%%%%%%%%%%%%%%%%%%%%%%%%%%%%%%%%%
 \section{Theory for linear models}
 \label{sec:linear_models}
%%%%%%%%%%%%%%%%%%%%%%%%%%%%%%%%%%%%%%%%%%%%%%%%%%%%%%%%%%%

Here, we develop an econometric framework for estimating spillover effects from sampled links when outcomes are linear in the (weighted) sum of neighbours' treatments (spillovers). In \ref{sec:setup}--\ref{sec:bias}, we introduce our setting. We show how unobserved spillovers may depend on sampled spillovers through the sampling rule. Regressing outcomes on sampled spillovers then yields biased estimates. To illustrate when this might happen, we provide a sufficient condition for unobserved spillovers to depend on sampled spillovers. Common sampling rules in applied research satisfy this condition.

Then, in \ref{sec:rescale}, we introduce our bias-corrected estimator, and show how researchers can construct it in practice. Our main result is that, under independence of treatment from links (as in a quasi-experimental or experimental setting) and plausible assumptions on the sampled network, researchers can construct a bias-corrected estimator from regression estimates and aggregate statistics of the degree distribution of the network. The statistics needed depend on the relationship between sampled and unobserved degrees generated through the sampling rule. Next, we derive the asymptotic distribution of these estimators, and present a bootstrap estimator for the variance. Finally, we show how researchers can use the results to assess robustness of estimators to sampling links. In the subsequent sections, we extend this approach to non-linear estimators and cases where treatment depends on links.

%%%%%%%%%%%%%%%%%%%%%%%%%%%%%%%%%%%%%%%%%%%%%%%%%%%%%%%%%%%%
\subsection{Setup}
\label{sec:setup}
%%%%%%%%%%%%%%%%%%%%%%%%%%%%%%%%%%%%%%%%%%%%%%%%%%%%%%%%%%%%%%
Let there be $\mathcal{N} = \{1, ..., N\}$ individuals situated on a simple network $\mathcal{G} = (\mathcal{N}, \mathcal{E}^{\mathcal{G}}, \mathcal{W}^{G})$ where $\mathcal{E}^{\mathcal{G}}$ is the set of edges, and $\mathcal{W}^{G}$ is the set of weights on those edges. Represent the relationships by the $N \times N$ adjacency matrix $G$, where elements $g_{ij} \in \{0,1\}$ if the network is unweighted and $g_{ij} \in \mathbb{R}$ if the network is weighted. The adjacency matrix is a single draw from some network-generating distribution $F_{G}$ that we leave general. Define the degree of individual $i$ as $d_{i} = \sum_{j}g_{ij}$ the (possibly weighted) number of connections from all other individuals to $i$.\\

 Instead of observing the true network, the researcher samples a set of edges and weights between individuals in $\mathcal{N}$ through the non-stochastic sampling rule $S:(\mathcal{E}^{G}, \mathcal{W}^{G}) \rightarrow (\mathcal{E}^{H}, \mathcal{W}^{H})$ such that $\mathcal{E}^{H} \cap \mathcal{E}^{G} \neq \emptyset$. Denote the \textit{sampled network} by $\mathcal{H} = (\mathcal{N}, \mathcal{E}^{\mathcal{H}}, \mathcal{W}^{H})$. To fix ideas, consider the three following examples.\\

 \textbf{Example -- classroom setting, fixed choice sampling rule.} 
 The individuals are children in a classroom. Links denote friendships between children, and any weights may denote time spent together. A child's degree is their number of friends.
 To collect network data, the researcher asks each individual to name at most $m$ friends. This sampling rule is commonly used to collect network data through surveys \citep{Coleman1957, Calvó-Armengol2009, Oster2012, Banerjee2013, Shakya2017}.\\

 \textbf{Example -- village setting, sampling based on group membership.} 
 The individuals are villagers across a set of villages. Links denote borrowing relationships between villagers, and any weights may denote amount lent to each other. A villager's degree is their number of individuals they have lent things to. To collect network data, the researchers assume that all individuals within the same village lend to each other. This is common in observational data where researchers can tell which types of individuals might be connected, but not the exact connections \citep[e.g][]{Chetty2011, Bloom2013, Carrell2013}. \\

 \textbf{Example -- firm supply network, high-weight links.} The individuals are firms. Links denote supply relationships, and any weights denote the proportion of total sales that goes to that firm. A firm's degree is their number of customers. To collect network data, researchers take the links where weights are greater than a threshold  -- the `most important' connections. This is common in observational data where individuals must disclose important interactions \citep{Atalay2011, Barrot2016}. For example, US publicly listed firms must disclose customers that make up at least $10\%$ of their sales to the Securities and Exchange Commission. \\

 We can split the adjacency matrix of the true network into a sampled part $H$ and an unsampled part $B$.

\begin{equation}
G = H + B.
\label{eq: split-adjacency}
\end{equation}

In our examples, $B$ encodes the network of unrecorded friendships, of villagers that do not actually lend to each other, and of firms' smaller customers. This decomposition yields a simple expression for estimator bias.

Let $\mathcal{B}$ denote the set of nodes with at least one (incoming) link sampled incorrectly.\footnote{Equivalently, $\mathcal{B}$ is the index set of rows of $B$ with at least one non-zero entry.} Further, define the sampled degree of node $i$  -- the total (weighted) number of sampled connections from all other individuals to $i$ -- as $d^{H}_{i} = \sum_{j}h_{ij}$. The unobserved degree of node $i$ -- the total (weighted) number of connections from all other individuals to $i$ that are not sampled -- is $d^{B}_{i} = \sum_{j} g_{ij} - \sum_{j}h_{ij}$. In our examples, these are the numbers of unobserved friends per child, unrecorded lending ties per villager, and undisclosed customers per firm.\\

Consider the problem of estimating the causal effect or structural parameter $\beta$ -- the `spillover effect' of an additional neighbour being treated on outcomes -- in the model

\begin{equation}
\label{eq:structural_model}
y_{i} = \beta \sum_{j}g_{ij}x_{j} + \epsilon_{i}.
\end{equation}

Outcomes $y_{i}$ are linear in the (weighted) sum of treatment $x_{i}$ of neighbours on the network (we refer to this sum as `spillovers').\footnote{Formally, $((g_{ij})_{j=1}^{N}, x_{i}, \epsilon_{i})_{i=1}^{N}$ can be described with some joint distribution that we do not restrict here.} Results apply to functional forms including an intercept, controls, and panel data (see Appendix A.3), as well as alternative specifications where researchers construct a dummy variable for at least one neighbour being treated \cite[e.g][see Appendix A.5]{Barrot2016}.

The researcher only observes the sampled network. So, they only observe the (weighted) sum of treatments of sampled neighbours

\begin{equation}
\sum_{j}h_{ij}x_{j} = \begin{cases}
\sum_{j}g_{ij}x_{j}  &\text{ if } i \notin \mathcal{B},\\
\sum_{j}g_{ij}x_{j} - \sum_{j}b_{ij}x_{j} &\text{ if } i \in \mathcal{B},
\end{cases}
\end{equation}

as opposed to those of the true neighbours. These sampled spillovers only equal true spillovers when the researcher samples all an individual's links correctly. Deviations from true spillovers depend on the sampling rule.\\

\textbf{Example -- classroom setting, fixed choice sampling rule.} Suppose the researchers sample at most $m$ friends per child. If a child has fewer than $m$ friends, the researcher samples all of their possible friendships correctly. If a child has more than $m$ friends, some are missed, and only $m$ links are recorded ($i \in \mathcal{B}$ if $d_{i} > m$). Therefore, the spillovers they sample are

\begin{equation*}
\sum_{j}h_{ij}x_{j} = \begin{cases}
\sum_{j}g_{ij}x_{j}  &\text{ if } d_{i} \leq m\\
\sum_{j}g_{ij}x_{j} - \sum_{j}b_{ij}x_{j} &\text{ otherwise}
\end{cases}
\end{equation*}

-- equal to true spillovers for children with fewer than $m$ friends, but different for children with more than $m$ friends. In expectation, the difference weakly increases with the number of friendships the child has.\\
%\begin{equation*}
%d_{i}^{B}  = \begin{cases}

%& \sum_{j=1}^{N}G^{*}_{ij} - m \text{ if } \sum_{j=1}^{N}G^{*}_{ij} > m,\\
%& 0 \text{ else. }
%\end{cases}
%\end{equation*}

\textbf{Example -- village setting, sampling based on group membership.} Consider a case where the villages are all size $m$, so researchers assume that all $m$ villagers within each village are connected. This adds links for villagers with fewer than $m$ neighbours ($i \in \mathcal{B}$ if $d_{i} < m$). Therefore, sampled spillovers are

\begin{equation*}
\sum_{j}h_{ij}x_{j} = \begin{cases}
\sum_{j}g_{ij}x_{j}  &\text{ if } d_{i} =  m\\
\sum_{j}g_{ij}x_{j} - \sum_{j}b_{ij}x_{j} &\text{ otherwise}
\end{cases}
\end{equation*}

-- equal to true spillovers for the villagers who lend to all  $m$ others in the village, but more than true spillovers for villagers with fewer connections. In expectation, the difference weakly decreases with the number of links the villager has.\\

\textbf{Example -- firm supply network, high-weight links.} Consider a case where researchers only sample links above some weight $\tau$. Unless all firm supply relationships have weight greater than $\tau$, researchers sample fewer links to firms than they actually have. Therefore, sampled spillovers are

\begin{equation*}
\sum_{j}h_{ij}x_{j} = \begin{cases}
\sum_{j}g_{ij}x_{j}  &\text{ if } g_{ij} >  \tau \text{ }\forall j\\
\sum_{j}g_{ij}x_{j} - \sum_{j}b_{ij}x_{j} &\text{ otherwise.}
\end{cases}
\end{equation*}\\

Assume that treatment is independently and identically distributed across nodes, and distributed independently of link strength.

\begin{assumption}\label{asm:distribution_treatment} Distribution of treatment $x_{i}$.
\begin{enumerate}[leftmargin=2cm, label={\textbf{\Alph*:}}, ref={1-\Alph*}]
\item \label{asm:distribution_treatment-A}%
\smallskip $x_{i} \sim \text{i.i.d. }F_{X}$ -- treatment is drawn i.i.d. from a common distribution,
\item \label{asm:distribution_treatment-B}%
$x_{j} \perp\!\!\!\perp g_{ij}, h_{ij} \text{ } \forall i,j \in \mathcal{N}$ -- treatment is distributed independently of true and sampled link strength.
\end{enumerate} 
\end{assumption}

Assumption \ref{asm:distribution_treatment-B} corresponds to experiments where the researcher directly assigns treatment \citep[e.g][]{Miguel2004, Oster2012, Conley2010}, natural experiments that assign treatment to some individuals on a network and not others \citep[e.g][]{Barrot2016, Carvalho2020}, or when treatment is determined by some other process unrelated to the network \citep[e.g][]{Coleman1957, Calvó-Armengol2009}. This assumption is central to constructing the bias-corrected estimators in Section \ref{sec:rescale}, but it is not required for earlier results. It fails when treatment is targeted by a planner based on network structure, or individuals can endogenously adjust links based on treatment. In section \ref{sec:extensions}, we discuss bias correction when assumption \ref{asm:distribution_treatment-B} does not hold. \\

Next, assume that (\ref{eq:structural_model}) is specified correctly

\begin{assumption}
\label{asm:structural_shocks}
Distribution of structural shocks. $E(\epsilon_{i}) = 0$, $\sum_{j}g_{ij}x_{j}, \sum_{j}h_{ij}x_{j} \perp\!\!\!\perp \epsilon_{i}$. 
\end{assumption}

Finally, assume that the expectation of the square of observed spillovers is finite.\footnote{ In the general case with an intercept, controls etc in Appendix A.3, this is the familiar assumption that regressors have finite variance \citep{Cameron2005}.}

\begin{assumption}
\label{asm:finite_variance}
    Finite second moment of observed spillovers. $E((\sum_{j}h_{ij}x_{j})^{2}) < \infty$. 
\end{assumption}

These are weak assumptions that justify the use of regression estimators for spillover effects. For asymptotic results, we assume standard regularity conditions on $\sum_{j}g_{ij}x_{j}$ that allow us to apply standard laws of large numbers and central limit theorems for independently but not identically distributed data, and a technical regularity condition on the dependence between observed and unobserved spillovers. For brevity, we list these in Appendix A.1.

%%%%%%%%%%%%%%%%%%%%%%%%%%%%%%%%%%%%%%%%%%%%%%%%%%%%%%%%%%%%%%%%%%
\subsection{Sampling generates endogeneity}
\label{sec:bias}
%%%%%%%%%%%%%%%%%%%%%%%%%%%%%%%%%%%%%%%%%%%%%%%%%%%%%%%%%%%%%%%%%

Suppose a researcher estimates spillover effects by regressing outcomes on sampled spillovers:
\begin{equation}
y_{i} = \beta \sum_{j}h_{ij}x_{j} + \xi_{i}.
\end{equation}

This regression model is misspecified. Outcomes depend upon all spillovers. But the researcher only includes sampled spillovers in their model. By using our decomposition of the adjacency matrix in (\ref{eq: split-adjacency}), we can express the misspecification in a very simple form. Substituting in the decomposition in (\ref{eq: split-adjacency}) we obtain: 

\begin{equation*}
\xi_{i} = \beta \sum_{j}b_{ij}x_{j} + \epsilon_{i}.
\end{equation*}

We see that by sampling links, the researcher inadvertently creates an omitted variable -- spillovers on unobserved links -- that enters the error term. If this covaries with sampled spillovers due to how the researcher samples the network, the estimator will be biased. Thus, the regression estimator is

\begin{equation}
\label{eq:naive_ols}
\hat{\beta}^{\text{OLS}} = \beta \Big(1+\frac{\frac{1}{N}\sum_{i} (\sum_{j}h_{ij}x_{j})(\sum_{j}b_{ij}x_{j})}{\frac{1}{N}\sum_{i} (\sum_{j}h_{ij}x_{j})^{2}} \Big) + \frac{\frac{1}{N}\sum_{i} (\sum_{j}h_{ij}x_{j})\epsilon_{i}}{\frac{1}{N}\sum_{i} (\sum_{j}h_{ij}x_{j})^{2}}.
\end{equation}

The estimated spillover effect equals the true spillover effect plus the dependence between sampled and unobserved spillovers times the true spillover effect.

If sampled spillovers $\sum_{j}h_{ij}x_{j}$ and unobserved spillovers $\sum_{j}b_{ij}x_{j}$ covary, then 

\begin{equation*}
E\Big(\frac{\frac{1}{N}\sum_{i} (\sum_{j}h_{ij}x_{j})(\sum_{j}b_{ij}x_{j})}{\frac{1}{N}\sum_{i} (\sum_{j}h_{ij}x_{j})^{2}} \Big) \neq 0.
\end{equation*} 

Therefore the estimator is biased.

\begin{proposition}
Make assumptions \ref{asm:distribution_treatment-A},  \ref{asm:structural_shocks}, \ref{asm:finite_variance}. If 

\begin{equation*}
E\Big(\frac{\frac{1}{N}\sum_{i} (\sum_{j}h_{ij}x_{j})(\sum_{j}b_{ij}x_{j})}{\frac{1}{N}\sum_{i} (\sum_{j}h_{ij}x_{j})^{2}} \Big) \neq 0,
\end{equation*}

then 
\begin{equation}
E(\hat{\beta}^{\text{OLS}}) = \beta \Big(1+ E\Big( \frac{\frac{1}{N}\sum_{i} (\sum_{j}h_{ij}x_{j})(\sum_{j}b_{ij}x_{j})}{\frac{1}{N}\sum_{i} (\sum_{j}h_{ij}x_{j})^{2}}\Big) \Big)\neq \beta.
\end{equation}
\end{proposition}

This tells us that estimates of spillover effects can be biased regardless of the treatment assignment if unobserved and spillovers covary due to how the researcher samples links.
%\paragraph{Links missing at random}

%Sampling errors can also generate dependence between observed and unobserved spillovers. Consider a case where researchers miss each true link at rate $q$. Then
%\begin{equation*}
%d_{i}^{B}  = \sum_{j}G^{*}_{ij}q, \text{ and }d_{i}^{G}  = \sum_{j}G^{*}_{ij}(1-q), 
%\end{equation*}
%which both depend on $G^{*}_{ij}$. So as an individual's true degree increases, both the mean number of true and missing links also increases. Therefore, $E(BX) \neq 0$, and $(BX)_{i}$ is positively related to $(GX)_{i}$.

Unlike attenuation bias from classical measurement error, estimates can be larger or smaller in magnitude than the true spillover effect. The estimator is upwards biased if $E(\frac{1}{N}\sum_{i} (\sum_{j}h_{ij}x_{j})(\sum_{j}b_{ij}x_{j})) > 0$, and downwards biased if $E(\frac{1}{N}\sum_{i} (\sum_{j}h_{ij}x_{j})(\sum_{j}b_{ij}x_{j})) < 0$. Adding controls only removes the bias when they capture the unobserved spillovers (see Appendix A.3). To illustrate, consider our running examples with a binary treatment $x_{i} \in {0,1}$.\\

\textbf{Example -- classroom setting, fixed choice sampling rule.} The error term contains the sum of treatments of additional friends on the children with more friends than the researcher samples. As this number is higher for the individuals with more friends, it covaries positively with sampled spillovers. Therefore the spillover estimate is upward biased. \\
%\begin{equation*}
%d_{i}^{B}  = \begin{cases}

%& \sum_{j=1}^{N}G^{*}_{ij} - m \text{ if } \sum_{j=1}^{N}G^{*}_{ij} > m,\\
%& 0 \text{ else. }
%\end{cases}
%\end{equation*}

\textbf{Example -- village setting, sampling based on group membership.} The error term subtracts the sum of treatments of the other individuals in the village that each individual does not lend to. The more individuals they actually lend to, the closer this number is to zero. As this number is lower for individuals with fewer friends, it covaries negatively with sampled spillovers. Therefore the spillover estimate is downward biased. \\

\textbf{Example -- firm supply network, high-weight links.} The error term contains the sum of treatments of additional customers of firms with link weights less than the threshold. Under the standard distribution used to model firm sales, this will covary positively with observed sales \citep[][]{Herskovic2020}.\footnote{See the corresponding simulation in section \ref{sec:simulations}.} Therefore the spillover estimate is upward biased. \\

An obvious next question is when sampling schemes induce dependence between spillovers on sampled and unobserved links. In other words, when does sampling link lead to biased spillover estimates? Suppose all links on the network have the same sign, Assumption \ref{asm:distribution_treatment} holds, and expected treatment is non-zero. A sufficient condition is that the expected number of unobserved links of each individual has the same sign -- so the researcher either samples a subset or superset of the true links.

\begin{proposition}
Make assumption \ref{asm:distribution_treatment-A}, \ref{asm:distribution_treatment-B}. Further, assume that all links on the network have the same sign -- either $g_{ij} \geq 0$ or $g_{ij} \leq 0$ $\forall j$ -- and that $E(x) \neq 0$. Then if the expectation of unobserved degree has the same sign for all nodes with potentially unsampled links

    \begin{equation*}
    E(d^{B}_{i}|d^{H}_{i}) \geq 0 \text{ }\ \forall i \in \mathcal{B} \text{ or }  E(d^{B}_{i}|d^{H}_{i}) \leq 0 \text{ }\ \forall i \in \mathcal{B}
    \end{equation*}

    and is non-zero for at least one $i$, then 
\begin{equation*}
E \Big(\frac{1}{N}\sum_{i} (\sum_{j}h_{ij}x_{j})(\sum_{j}b_{ij}x_{j}) \Big) \neq 0.
\end{equation*}
    
\end{proposition}

 This covers the sampling schemes commonly used to study economic and social networks, including the examples of fixed choice designs, group membership, and sampling high weight links given above. Many social and economic networks have links with all positive or all negative signs, such as firm-level production networks \citep{Atalay2011}, information sharing networks \citep{Banerjee2013}, and friendship networks \citep{Calvó-Armengol2009}. For intuition, we give an extended example with a fixed choice design in Appendix A.2.

%%%%%%%%%%%%%%%%%%%%%%%%%%%%%%%%%%%%%%%%%%%%%%%%%%%%%%%%%%%%%%
\subsection{Bias-corrected estimators}
\label{sec:rescale}
%%%%%%%%%%%%%%%%%%%%%%%%%%%%%%%%%%%%%%%%%%%%%%%%%%%%%%%%%%%%%%

We can write a bias function for the linear regression estimator \citep{MacKinnon1998}

\begin{equation*}
\hat{\beta}^{\text{OLS}} = \beta + \beta \frac{\frac{1}{N}\sum_{i} (\sum_{j}h_{ij}x_{j})(\sum_{j}b_{ij}x_{j})}{\frac{1}{N}\sum_{i} (\sum_{j}h_{ij}x_{j})^{2}} + \frac{\frac{1}{N}\sum_{i} (\sum_{j}h_{ij}x_{j})\epsilon_{i}}{\frac{1}{N}\sum_{i} (\sum_{j}h_{ij}x_{j})^{2}}.
\end{equation*}

Taking expectations and solving for $\beta$ gives us a bias-corrected estimator\footnote{
This approach is equivalent to controlling for the expected unsampled spillovers amongst nodes that have at least some incorrectly sampled links

\begin{equation*}
z_{i} = \begin{cases}
0 &\text{ if } i \notin \mathcal{B},\\
E(\sum_{j}h_{ij}x_{j}|i \in \mathcal{B}) &\text{ if } i \in \mathcal{B}.
\end{cases}
\end{equation*}

But it does not require knowing which nodes have some incorrectly sampled links, just how many. In many cases -- such as the group membership and high-weight link examples -- the researcher does not know which nodes have some incorrectly sampled links. Thus, we consider the bias-corrected estimator instead.}

\begin{theorem}
Define $\eta = E \Big(\frac{\frac{1}{N}\sum_{i} (\sum_{j}h_{ij}x_{j})(\sum_{j}b_{ij}x_{j})}{\frac{1}{N}\sum_{i} (\sum_{j}h_{ij}x_{j})^{2}} \Big)$. Make Assumptions \ref{asm:distribution_treatment-A},  \ref{asm:structural_shocks}, \ref{asm:finite_variance}. The estimator

\begin{equation}
\label{eq:estimator_true_eta}
\hat{\beta} = \frac{\hat{\beta}^{\text{OLS}}}{1+ \eta}
\end{equation}

is an unbiased estimator of $\beta$ i.e $E(\hat{\beta}) = \beta$.
\end{theorem}

In words, the researcher needs to rescale their estimate of the spillover effect to adjust for the dependence between sampled spillovers and unobserved spillovers induced by the sampling rule. Because this dependence is unobserved by construction, in practice the researcher needs to compute

\begin{equation*}
\label{eq:estimator_true_eta}
\hat{\beta} = \frac{\hat{\beta}^{\text{OLS}}}{1+ \hat{\eta}}
\end{equation*}

where $\hat{\eta}$ is an estimate $\eta$.

Now, focus on cases where shocks are distributed independently of links (Assumption \ref{asm:distribution_treatment-B}). As discussed earlier, these include experiments assigning treatment across networks, quasi-experimental designs, and cases where the network is fixed before some untargeted treatment. In these cases, the researcher can construct a good approximation $\hat{\eta}$, and therefore estimate of the spillover effect, from aggregate statistics of the degree distribution and expected treatment. To see this, consider the Taylor expansion of $\eta$ around the mean observed and unobserved spillovers \citep{Billingsley2012}

\begin{align}
\label{eq:Taylor}
\eta &= \frac{E(\frac{1}{N}\sum_{i} (\sum_{j}h_{ij}x_{j})(\sum_{j}b_{ij}x_{j}))}{E(\frac{1}{N}\sum_{i} (\sum_{j}h_{ij}x_{j})^{2})} + \mathcal{O}(\frac{1}{\frac{1}{N}\sum_{i} (\sum_{j}h_{ij}x_{j})^{4}}), \notag\\
&\approx \frac{E(\frac{1}{N}\sum_{i} (\sum_{j}h_{ij}x_{j})(\sum_{j}b_{ij}x_{j}))}{E(\frac{1}{N}\sum_{i} (\sum_{j}h_{ij}x_{j})^{2})}.
\end{align}

The remainder term  $\mathcal{O}(\frac{1}{\frac{1}{N}\sum_{i} (\sum_{j}h_{ij}x_{j})^{4}})$ is negligible in most cases. If not, researchers can apply a higher order expansion or approximate the full expectation by simulation. 

Next, to compute the numerator, we need to impose an assumption on the distribution of observed and unobserved links. First, assume that the distribution of observed degree is independent of the distribution of unobserved degree amongst nodes with some potentially incorrectly sampled links.

\edef\oldassumption{\the\numexpr\value{assumption}+1}

\setcounter{assumption}{0}
\renewcommand{\theassumption}{\oldassumption.\alph{assumption}}

\begin{assumption}
\label{asm:indep_sampled_unsampled}
Distribution of unsampled degree -- $d^{B}_{i} \perp\!\!\!\perp  d^{H}_{i} |i \in \mathcal{B}$.
\end{assumption}

This assumption applies to many sampling schemes used in economic research when the underlying network is binary. For illustration, consider the following examples.\\

\textbf{Example -- classroom setting, fixed choice design with binary network.} If there are potentially unsampled friendships to a child $i \in \mathcal{B}$, we know that the sampled (in)degree equals the threshold value $d_{i}^{H} = m$ i.e they have $m$ friends. Therefore, the distribution of sampled degrees $d^{H}_{i}$ given that $i \in \mathcal{B}$ has a point mass at $m$. All children with some unsampled friendships have $m$ sampled friends. It follows that the distribution of the number of unsampled links $d^{B}_{i}$ is independent of the distribution of sampled links amongst individuals where $i \in \mathcal{B}$.\\

\textbf{Example -- village setting, group membership with binary network.} Assume for simplicity that all villages have an equal size $m$. For all $i$, the number of sampled neighbours equals one minus the village size $d_{i}^{H} = m-1$ by construction. Therefore, the distribution of the degree $d_{i}^{H}$ given that $i \in \mathcal{B}$ has a point mass at $m-1$. It follows that the distribution of the unsampled degree $d^{B}_{i}$ is independent of the distribution of sampled links amongst individuals where $i \in \mathcal{B}$.\footnote{This argument extends to the case where the size of the group varies across some groups, as long as the degree of each individual within each group does not itself depend on the size of the group.}\\

If no links differ in strength, we need not worry that subjects report links in an order that might violate this assumption.

Evaluating the approximation for $\eta$ under \ref{asm:indep_sampled_unsampled}, we can characterise the expected dependence in terms of: the mean sampled degree of nodes that have at least one potentially unsampled link, the mean missing degree of nodes that have at least one potentially unsampled link, and the expected treatment status of each node. Define

\begin{equation*}
\begin{aligned} 
\hat{d}^{H} &= \frac{1}{\sum_{i \in \mathcal{B}}1_{i}}\sum_{i\in \mathcal{B}_{i},j}h_{ij},    & \hat{d}^{B} &= \frac{1}{\sum_{i \in \mathcal{B}}1_{i}}\sum_{i \in \mathcal{B},j} b_{ij}, \\ 
\bar{x} &= \frac{1}{N}\sum_{i}x_{i}, & N^{B} &= |\mathcal{B}|.
\end{aligned}
\end{equation*}

Then, using (\ref{eq:Taylor}):

\begin{equation*}
\hat{\eta} \approx \frac{ \frac{N^{B}}{N} \hat{d}^{H}\hat{d}^{B} \bar{x}^{2}}{\frac{1}{N} \sum_{i}(\sum_{j}h_{ij}x_{j})^{2}},
\end{equation*}

\begin{proposition}
Make Assumptions \ref{asm:distribution_treatment-A}, \ref{asm:distribution_treatment-B} \ref{asm:structural_shocks}, \ref{asm:finite_variance}, \ref{asm:indep_sampled_unsampled}. Consider the estimator    
\begin{equation}
\hat{\beta} = \frac{\hat{\beta}^{\text{OLS}}}{1+ \hat{\eta}} \text{ where } \hat{\eta} = \frac{ \frac{N^{B}}{N} \hat{d}^{H}\hat{d}^{B}  \bar{x}^{2}}{\frac{1}{N} \sum_{i}(\sum_{j}h_{ij}x_{j})^{2}}
\label{eq:rescaled_estimator_1}
\end{equation}

Let $\hat{\beta}_{N}$ denote an estimate from sample size $N$. $E(\hat{\beta}_{N}) \approx \beta$ and $\hat{\beta} \xrightarrow[]{p} \beta$.

\end{proposition}
The rescaling factor $\hat{\eta}$ depends only on two aggregate network statistics -- the sampled mean degree and true mean degree of individuals who have at least one potentially unsampled link. It does not require knowing which specific links are missing. 

The requirement to obtain mean degrees for subsets of the population is relatively mild compared to conditioning on unobservable counterfactual networks \citep{Breza2020, Herstad2023, Borusyak2023} or constructing multiple network measures \citep{Lewbel2023}. As they are aggregate network statistics, they are relatively easy to collect or estimate. In a survey, the researcher could collect both by including one more question: `How many of these types of connections do you have?'. Data providers can disclose them while preserving privacy. In cases where the researcher cannot sample individuals in the network -- for example when using data collected by others -- researchers can use the statistics from similar, better-sampled  networks. Additional survey questions could also help estimate the mean missing degree under relatively weak assumptions. For example, a researcher could use the question "How many of your friends smoke?" plus an assumption on the distribution of smokers in the population to recover mean missing degree in a friendship network. 

We can relax Assumption \ref{asm:indep_sampled_unsampled} and allow the number or strength of unobserved links to depend on the number or strength of observed links. For example, individuals may name stronger connections first in a survey. Instead, we adopt the weaker assumption that observed and unobserved link counts or strengths share a common conditional distribution.

\begin{assumption}
\label{asm:cond_sampled_unsampled}
There exists a joint distribution over $(d^{H}_{i}, d_{i})_{i \in \mathcal{B}}$ such that we can write $E(d^{B}_{i}|d^{H}_{i}) =E(d_{i}|d^{H}=d^{H}_{i}) - d^{H}_{i} \text{ } \forall i \in \mathcal{B}$.
\end{assumption}

\textbf{Example -- classroom setting, fixed choice design naming stronger connections first.} 
Assume that weighted degrees (number of friends) are drawn from a common degree distribution $F_{d}$, and for simplicity assume that all weights are positive. The researcher samples up to $m$ friendships per child. Children list their strongest friendships first. The strength of each unobserved friendship must be less than or equal to the lowest strength of the observed friendships. Otherwise, the child would have named that friendship earlier. So, for children with at least one potentially missing friendship, $0 \leq d_{i} - d^{H}_{i} \leq (N-m) \operatorname{min}\{h_{ij}|h_{ij} > 0\}$. Accordingly, $E(d^{B}_{i}|d^{H}_{i}) =  E(d_{i}|d_{i} \geq (N-m) \operatorname{min}\{h_{ij}|h_{ij} > 0\}) - d^{H}_{i}$.  \\

Denote the average number of unobserved links for individuals with some potentially unsampled links and $d^{H}$ sampled links as

\begin{equation*}
\hat{d}^{B}(d^{H}) = \frac{1}{\sum_{i \in \mathbf{B}} \mathbf{1}(d^{H}_{i} = d^{H})}\sum_{i \in \mathcal{B}} d^{B}_{i}\mathbf{1}(d^{H}_{i} = d^{H}).
\end{equation*}

Then, we approximate
\begin{equation*}
\hat{\eta} \approx \frac{ \frac{1}{N} \sum_{i \in \mathcal{B}} d^{H}_{i}\hat{d}^{B}(d_{i}^{H})\bar{x}^{2}}{\frac{1}{N} \sum_{i}(\sum_{j}h_{ij}x_{j})^{2}}
\end{equation*}

\begin{proposition}
Make Assumptions \ref{asm:distribution_treatment-A}, \ref{asm:distribution_treatment-B},  \ref{asm:structural_shocks}, \ref{asm:finite_variance}, \ref{asm:cond_sampled_unsampled}. Consider the estimator
\begin{equation*}
\hat{\beta} = \frac{\hat{\beta}^{\text{OLS}}}{1+ \hat{\eta}} \text{ where } \hat{\eta} = \frac{ \frac{1}{N} \sum_{i \in \mathcal{B}} d^{H}_{i}\hat{d}^{B}(d_{i}^{H})\bar{x}^{2}}{\frac{1}{N} \sum_{i}(\sum_{j}h_{ij}x_{j})^{2}}
\end{equation*}

Letting $\hat{\beta}_{N}$ denote an estimate from a sample of size $N$, $E(\hat{\beta}_{N}) \approx \beta$ and $\hat{\beta} \xrightarrow[]{p} \beta$.
\end{proposition}

As before, constructing unbiased estimates only requires knowing aggregate network statistics, rather than which links are missing. The researcher can construct an estimate of the distribution of missing degree given sampled degree using the empirical distribution of (the strength of) total missing links conditional on the sampled links. As under Assumption \ref{asm:indep_sampled_unsampled}, these statistics can be collected by adding an additional question to a survey. Data providers can disclose these aggregate quantities without violating privacy. Researchers can also approximate them from the degree distribution of similar fully sampled networks together with knowledge of the sampling rule.

\let\theassumption\origtheassumption

\setcounter{assumption}{4}

%%%%%%%%%%%%%%%%%%%%%%%%%%%%%%%%%%%%%%%%%%%%%%%
\subsection{Asymptotic distribution}
%%%%%%%%%%%%%%%%%%%%%%%%%%%%%%%%%%%%%%%%%%%%%%%

Next, we characterise the asymptotic distribution of our estimator. First, consider the case where a value of $\eta$ is known. For example, a data provider might disclose it. For sample size $N$, define

\begin{equation*}
\eta_{N} = \frac{\frac{1}{N}\sum_{i} (\sum_{j}h_{ij}x_{j})(\sum_{j}b_{ij}x_{j})}{\frac{1}{N}\sum_{i} (\sum_{j}h_{ij}x_{j})^{2}}
\end{equation*}

Then, as is standard, the bias-corrected estimator is consistent and asymptotically normal \citep{Cameron2005}.
\begin{proposition}
    Make Assumptions 1,2,3, 5, A1. Then

    \begin{equation*}
\sqrt{N}(\hat{\beta} - \beta) \xrightarrow[]{d} N(0, \frac{1}{(1+\eta)^{2}} \Omega).
\end{equation*}
\end{proposition}

 Next, consider the more interesting case where we instead have to estimate $\hat{\eta}(\hat{\theta})$ as a function of some finite vector of parameters $\theta$ defined as the solution to the moment conditions

\begin{equation*}
\theta - \frac{1}N{}\sum_{i=1}^{N}\theta_{i} = 0.
\end{equation*}

In our examples above, $\theta$ are mean unobserved degrees. In this case, the asymptotic distribution of the spillover estimate depends upon both the uncertainty in the estimates of $\theta$ and the sensitivity of $\hat{\eta}(\hat{\theta})$ to $\hat{\theta}$ \citep{Newey1984}. 

\begin{proposition}
    Make Assumptions 1A ,2,3, 5, A1. Define 

\begin{align*}
\begin{pmatrix}
h_{1}(\theta)\\
h_{2}(\theta, \beta)
\end{pmatrix} &=
\begin{pmatrix}
\theta - \frac{1}N{}\sum_{i=1}^{N}\theta_{i}\\
\frac{1}{N}\sum_{i}(\sum_{j}h_{ij}x_{j})(y_{i} - (1+\eta(\theta))\beta(\sum_{j}h_{ij}x_{j}))
\end{pmatrix}\\
\begin{pmatrix}
K_{11} & K_{12}\\
K_{21} & K_{22}
\end{pmatrix}
&= \plim \frac{1}{N} \sum_{i} E
\begin{pmatrix}
-1& 0\\
-  \frac{\partial \eta(\theta)}{\partial \theta}\beta (\sum_{j}h_{ij}x_{j})^{2}  & -(1+\eta(\theta))(\sum_{j}h_{ij}x_{j})^{2}  
\end{pmatrix}\\
\begin{pmatrix}
S_{11} & S_{12}\\
S_{21} & S_{22}
\end{pmatrix} &=
\plim \frac{1}{N} \sum_{i} E \begin{pmatrix}
h_{1i}h_{1i}'& h_{2i}h_{1i}'\\
h_{2i}h_{1i}' & h_{2i}h_{2i}'
\end{pmatrix}
\end{align*}

$\hat{\beta}$ is a consistent estimator of $\beta$. Furthermore,

    \begin{equation*}
\sqrt{N}(\hat{\beta} - \beta) \sim N(0, K_{22}^{-1}(S_{22} + K_{21}K_{11}^{-1}S_{11}K_{11}^{-1}K_{21}' - K_{21} K^{-1}_{11}S_{12}  -S_{21}K^{-1}_{11}K_{21}') K_{22}^{-1}).
\end{equation*}

\end{proposition}

In practice, we propose using a bootstrap to estimate the variance of the estimator. For example, assume that we are computing $\hat{\eta}$ under Assumption \ref{asm:indep_sampled_unsampled}. In the first step, we simulate $P$ different possible unobserved graphs consistent with the same missing degree. In the absence of any link function that determines how likely any two individuals are to be connected given that their links are not sampled correctly, we assume that incorrectly observed links are distributed uniformly at random over all possible missing entries in $B$. In the second step, we construct $M$ bootstrap estimates of $\hat{\beta}$ for each $B$. Similar bootstrap estimators can be derived under alternative assumptions on the network sampling process. 

\begin{algorithm}
\label{alg:net_lasso}
  \caption{Bootstrap estimator for $\hat{s}(\hat{\beta})$ under \ref{asm:indep_sampled_unsampled}}
  \begin{algorithmic}[1]
    \Procedure{Bootstrap }{$d^{B}$, $H$,  $\{\mathcal{E}^{\mathcal{H}}_{i}\}_{i=1}^{N}$, $x$, $y$}
    \For{$j \in 1, ..., M$}
    \State \textbf{Draw} $\{B_{ik}| k \notin \mathcal{E}^{\mathcal{H}}_{i} \} s.t \sum_{\{B_{ik}| k \notin \mathcal{E}^{\mathcal{H}}_{i} \}}B_{ik} = N \bar{d}^{B}.$
    \State Construct $\{\hat{\beta_{kj}}\}_{k=1}^{P}$ by a regression bootstrap from $B^{j}, H, x, y$.
    \EndFor
    \State $\bar{\beta}_{kj} = \frac{1}{MP}\sum_{k,j}\hat{\beta}_{kj}$.
    \State $\hat{s}(\hat{\beta}) = \sqrt{ \frac{1}{MP}\sum_{k,j}(\hat{\beta}_{kj} - \bar{\beta}_{kj})^{2}} $
    \EndProcedure
  \end{algorithmic}
\end{algorithm}

\subsection{Robustness}
\

In addition to constructing bias-corrected estimators, the researcher can use Theorem 1 to assess the robustness of spillover estimates to sampling bias in two ways. First, they can recover the value of $\eta$ needed to reduce the spillover estimate below some decision threshold $\tau$. Examples include thresholds relevant for optimal policy decisions or values required for test statistics to cross critical values at preferred significance levels.

\begin{proposition}

Make Assumptions \ref{asm:distribution_treatment-A},  \ref{asm:structural_shocks}, \ref{asm:finite_variance}. Then
    \begin{align*}
\beta &> \tau \text{ if and only if}\\
\eta &< \frac{\hat{\beta}^{\text{OLS}} - \tau}{\tau}.
    \end{align*}
\end{proposition}

 Second, if the researcher can bound the dependence of observed and unobserved spillovers $\eta \in [\eta_{\text{min}}, \eta_{\text{max}}]$, then the true spillover effect is bounded as

\begin{equation*}
\beta \in \Big[\frac{\hat{\beta}^{\text{ OLS}}}{1 + \eta_{\text{max}}}, \frac{\hat{\beta}^{\text{ OLS}}}{1 + \eta_{\text{min}}}\Big].
\end{equation*}

Under assumption 4.a, these results depend only on the mean missing degree among individuals with at least one missing link. In this case, the decision threshold can be rewritten as a function of the mean number of missing links among individuals with at least one missing link:

\begin{align}
&\hat{\beta}^{\text{ OLS}} > \tau \text{ if and only if} \notag \\
&\hat{d}^{B} < \Big(\frac{\frac{1}{N}\sum_{i}(\sum_{j}h_{ij}x_{j})^{2}}{\frac{N^{H}}{N}\bar{x}^{2}\hat{d}^{H}}\Big) \frac{\hat{\beta}^{\text{ OLS}} - \tau}{\tau}.
\label{eq:robustness_degree}
\end{align}

Thus, spillover effects exceed a threshold if and only if the researcher is missing fewer than a certain number of links. Moreover, the bounds depend on the minimum and maximum mean number of missing links

\begin{equation*}
\eta_{max} = \eta(\hat{d}^{B}_{max}), \text{ }\eta_{min} = \eta(\hat{d}^{B}_{min}).
\end{equation*}

%%%%%%%%%%%%%%%%%%%%%%%%%%%%%%%%%%%%%%%%%%%%%%%%%%%%%%%%%
\section{Theory for nonlinear social network models}
 \label{sec:nonlinear_models}
%%%%%%%%%%%%%%%%%%%%%%%%%%%%%%%%%%%%%%%%%55

We can also extend our bias-correction approach to models where outcomes are linear in spillovers of indirect neighbours. To show this, we consider a non-linear specification commonly used in research on social networks \citep{Bramoulle2009, Calvó-Armengol2009}. Here, sampling bias also affects the standard instruments used to account for endogeneity in lagged spillovers. Thus, researchers must both correct instruments and bias-correct the resulting estimates.

\subsection{Setup}

An alternative model often used to measure spillover effects specifies outcomes as linear in the sum of indirect spillovers across all paths through the network, rather than just direct spillovers \citep[e.g][]{Calvó-Armengol2009, Carvalho2020}. Formally

\begin{align}
\label{eq:sar_dgp}
 y &= \lambda Gy + x\beta + \epsilon\\
 &=(I-\lambda G)^{-1}(x \beta + \epsilon). \notag
\end{align}

where $y = (y_{1}, y_{2}, ..., y_{n})$ is the $N \times 1$ vector of individual outcomes, and $x = (x_{1}, ..., x_{n})$ is the $N \times 1$ vector of treatments.\footnote{Without loss of generality, we focus on the case without contextual effects $Gx$ here for ease. Our results extend to estimates of contextual spillover effects. Then, researchers also need to account for the identification problems raised in \cite{Manski1990, Blume2015}. } The inverse

\begin{equation*}
(I-\lambda G)^{-1} = \sum_{k=1}^{\infty}\gamma^{k}G^{k}
\end{equation*}

sums spillovers across all paths of length $k = 1,2, ...$ through the network. 

In this setting, sampling the network generates more complex misspecification than in the linear model. Comparing the true paths of length $k$ to sampled paths of length $k$ using our decomposition (\ref{eq: split-adjacency}) gives

\begin{align*}
G^{k} &= (H + B)^{k}\\
&= H^{k} + H^{k-1}B + ... + B^{k}.
\end{align*}

True paths include paths through only sampled links, paths through only unobserved links, and paths created by combining sampled and unobserved links. Estimator bias therefore depends on the covariance between treatment transmitted along sampled paths and treatment transmitted along additional paths

A researcher estimates structural parameters $\beta, \lambda$ -- the effect of treatment on outcomes, and the spillover effect of one individual's outcomes on others' -- using the sampled network by fitting

\begin{equation}
\label{eq:sar_dgp}
 y = \lambda Hy + x\beta + \xi.
\end{equation}

Using our decomposition (\ref{eq: split-adjacency}), we see that by sampling the network the researcher creates an omitted variable $By$ that enters the error term 

\begin{equation*}
\xi = \lambda By + \epsilon.
\end{equation*}

The standard approach is to estimate this model by two-stage least squares, constructing instruments using the treatment of sampled friends of sampled friends \citep{Bramoulle2009}. We focus on this estimator, rather than the maximum-likelihood estimator. We adopt the standard assumptions used for this estimator \citep{Kelejian1998, Bramoulle2009, Blume2015}, spelled out in Appendix A.7. Denote our regressors as $z^{*} = \begin{pmatrix} Gy, x \end{pmatrix}$, $z = \begin{pmatrix} Hy, x \end{pmatrix}$. Call $z_{B} = z^{*} - z = \begin{pmatrix} By, 0 \end{pmatrix}$, and denote instruments built from the sampled network as $J = H(I-H)^{-1}x =  \begin{pmatrix} Hx & H^{2}x & ... \end{pmatrix}$, and the corresponding projection matrix as $P_{J}$. The two-stage least squares estimator is
\begin{equation*}
    \begin{pmatrix} 
    \hat{\lambda}^{\text{ 2SLS}}\\
    \hat{\beta}^{\text{ 2SLS}}
    \end{pmatrix} = (z^{'} P_{J}z)^{-1}z^{'} P_{J}y.
    \end{equation*}
    
As with the linear model, this estimator can be biased by sampling. 

\begin{proposition}
Make assumption 2 and the standard assumption in A.7. Let $P$ denote a projection matrix, $z = \begin{pmatrix} Gy, x \end{pmatrix}$, $J = \begin{pmatrix} x, Hx, H'Hx, .... \end{pmatrix}$. There exist $H, B$ such that the two-stage least-squares estimator 
    \begin{equation*}
    \hat{\theta}^{\text{ 2SLS}} = 
    \begin{pmatrix} 
    \hat{\lambda}^{\text{ 2SLS}}\\
    \hat{\beta}^{\text{ 2SLS}}
    \end{pmatrix} = (z^{'} P_{J}z)^{-1}z^{'} P_{J}Y.
    \end{equation*}
    is biased and inconsistent.
\end{proposition}

To see why, note that the instrument exogeneity condition is
\begin{align*}
E(J'\xi) &= E\Big(\begin{pmatrix} x, Hx, H'Hx, .... \end{pmatrix}' (\lambda By + \epsilon)\Big)\\
&= E\Big(\begin{pmatrix} x, Hx, H'Hx, .... \end{pmatrix}' (\lambda By)\Big) + E\Big(\begin{pmatrix} x, Hx, H'Hx, .... \end{pmatrix}' \epsilon)\Big).
\end{align*}

The second term is the instrument exogeneity condition if $H$ is the true network. Expanding the first term gives

\begin{equation*}
E\Big(\begin{pmatrix} x, Hx, H'Hx, .... \end{pmatrix}' (\lambda By)\Big) = E \Big((H(I-\lambda H)^{-1}x) (\lambda B(I-\lambda (H+B))^{-1}(x\beta + \epsilon)) \Big),
\end{equation*}

the product of spillovers between two individuals on paths only containing sampled links, and spillovers between two individuals on paths containing either unsampled links alone or both sampled links and unsampled links. The estimator fails when these two covary.\footnote{Here, we make no assumption on the fraction of links that are incorrectly sampled. \cite{Lewbel2024} show that, in this setting, if the fraction of links that are incorrectly sampled falls quadratically in sample size, the two-stage least-squares estimator remains consistent. The first term in our instrument exogeneity condition vanishes as the sample size becomes larger. But, for the common sampling schemes listed above we would not expect the fraction of links incorrectly sampled to fall with sample size. Additionally, we see large finite-sample biases in simulations of common sampling schemes on networks.}

\subsection{Bias-corrected estimator}

We construct bias-corrected estimators using the logic in Section \ref{sec:linear_models}. First, we must construct correct instruments for spillovers through the sampled network. Once instruments are correct, the second stage is a linear regression of instrumented spillovers on the sampled network on outcomes. So, we can apply the same bias correction to the estimated coefficient in the second stage regression to account for the omitted $By$ term.

To construct correct instruments for $Hy$, we account for the expected number of missing paths between individuals. Following \citep{Kelejian1998}, we use that

\begin{align*}
E(Hy) &= E(H(I+\lambda(H+B))^{-1} x\beta)\\
&= E(H(H+B + (H+B)^2 + (H+B)^3 + ... )x\beta).
\end{align*}

Therefore we use instruments

\begin{equation*}
J^{*} =  \begin{pmatrix} Hx, & E(Bx), & E(HBx|H), & ...\end{pmatrix}
\end{equation*}

Implementation requires computing the expected number of unobserved paths of length $k$ between nodes through the network given the sampled network, using knowledge of the sampling scheme and an assumption on the distribution of missing links given observed links. To give an example, make the following assumption on the distribution of missing links.

\begin{assumption}
     The distribution of unobserved links is independent of the distribution of observed links for individuals with at least some unobserved links -- $B_{ij} \perp\!\!\!\perp  H_{jk} \text{ }\forall{i} \in \mathcal{B}, H_{ij} \perp\!\!\!\perp  B_{jk} \text{ }\forall{j} \in \mathcal{B}.$
\label{asm:indep_sampled_unsampled_links}
\end{assumption}

This assumption applies for networks under the common sampling schemes given above when all links are drawn from a common distribution. Other assumptions may be needed if, for example, some individuals are systematically more popular or name links in an order.

Expected numbers of walks then depend on sampled walks and powers of the mean missing degree. Considering the case of paths of length $2$ for simplicity, and imagining that there are $m$ possible incorrect entries in column $j$ of $H$, we have

\begin{align*}
E(HB|H)_{ik} &= E(\sum_{j}H_{ij}B_{jk}|H),\\
&= \sum_{j} E(H_{ij}B_{jk}|H) \text{ by linearity of E,}\\
&= \sum_{j}H_{ij}E(B_{jk}|H) \text{ by \ref{asm:indep_sampled_unsampled_links}},\\
&= \sum_{j}H_{ij} \frac{d^{B}_{j}}{ |\mathcal{N}| - m}.
\end{align*}

Then the researcher can proxy the numbers of missing paths through the network $H^{k-1}B, ...$ with the expected number of missing paths through the network given the sampled adjacency matrix and missing mean degree. Formally, this gives

\begin{proposition}
Under Assumption \ref{asm:indep_sampled_unsampled_links}, the variables $J^{*} = \begin{pmatrix}Hx, & d^{B}Hx, & H^{2}x, & ... \end{pmatrix}$ are valid instruments for $Hy$ conditional on $By$.
\end{proposition} 

The endogeneity problem from the missing $By$ in the second stage regression remains

\begin{align*}
y &= \lambda \widehat{Hy} + \xi,\\
\xi &= By + \epsilon.
\end{align*}

As in \ref{sec:linear_models}, the researcher can bias-correct estimates to deal with the omitted term $By$.

\begin{proposition}
Define 
\begin{equation*}
\hat{\theta}^{SS} = (z^{'} P_{J^{*}}z)^{-1}z^{'} P_{J^{*}}y \text{, }
\hat{Z} = P_{J^{*}}z\text{, }
\eta =(N^{-1}z'P_{J*}z)^{-1}N^{-1}\hat{z}'z_{B}.
\end{equation*}

The estimator
\begin{equation}
\hat{\theta} = (I+\eta)^{-1}\hat{\theta}^{SS}
\end{equation}
is an unbiased estimator of $\theta = \begin{pmatrix} \lambda\\ \beta \end{pmatrix}$.
\end{proposition}

This estimator is also consistent, and asymptotically normal. We derive these results in appendix A.7, and show in Appendix A.8 that the estimator performs well in finite samples by simulation.

%%%%%%%%%%%%%%%%%%%%%%%%%%%%%%%%%%%%%%%%%%%%%%%%%%%%%%%%%
\section{Extension -- treatment dependent on network structure}
\label{sec:extensions}
%%%%%%%%%%%%%%%%%%%%%%%%%%%%%%%%%%%%%%%%%55

In some cases, researchers may wish to estimate spillover effects when treatment depends on links. For example, treatment may be targeted by a planner based on network structure \citep[e.g][]{Galeotti2020}, or individuals may form their links based on treatment status \citep[for examples, see][]{Calvó-Armengol2009, Jackson2010}. 

Assume that Assumptions 1.A, 2,3 hold and the data is drawn from (\ref{eq:structural_model}) as before. Then, we can still construct estimates of the spillover effect by rescaling using Theorem 1. But we can no longer model $\hat{\eta}$ using purely aggregate statistics of the degree distribution under assumption 1.B. Instead, we need to model the distribution of unobserved links and treatment. Formally, unobserved links $b_{ij}$ are dependent on $x_{j}$,  and some additional dependence parameters $\theta$. Thus,

\begin{equation*}
\eta \approx \frac{E(\frac{1}{N}\sum_{i} (\sum_{j}h_{ij}x_{j})(\sum_{j}b_{ij}(x_{i},x_{j}, \theta)x_{j}))}{E(\frac{1}{N}\sum_{i} (\sum_{j}h_{ij}x_{j})^{2})}.
\end{equation*}

To compute $\hat{\eta}$, we need a way of modeling the expected treatment of the observed and unobserved neighbours given observed and unobserved links. 

One possible route is to fit a parametric model for the joint distribution of links on the network and treatment as in \cite{Borusyak2023} and \cite{Herstad2023}. Instead, we consider the case where the researcher does not want to impose a parametric model for joint distribution of treatment and links ex-ante, but they are willing to model treatment as dependent upon some network statistic. This is a weaker assumption, as the researcher does not have to place restrictions upon many other features of the network, as in a full parametric model. We propose using a copula, as copulas allow us to flexibly model the dependence structure between two distributions preserving their marginal distributions.

Here, for simplicity, assume that treatment $x_{j}$ depends upon in-degree $d_{j}$. Denote the observed distribution of treatment as $F_{X}$, and the distribution of the in-degree as $F_{D}$. The pairs $(x_{i}, d_{i})$ are distributed according to some unknown joint density function $G()$ with marginal distributions $F_{X}, F_{D}$. The researcher can flexibly model the joint density of treatment and this network statistic from empirical marginal distributions using a copula \citep{Nelsen2006, Trivedi2007}. 

\begin{definition}
A bivariate copula is a quasi-monotone function $C()$ on the unit square $[0,1] \times [0,1] \rightarrow [0,1]$ such that there exists some $a_{1}, a_{2}$ such that $C(a_{1}, y) = C(x, a_{2})\text{, and } C(1, y) = y, C(x, 1) = x$ $ \forall x,y \in [0,1]$.
\end{definition}

From Sklar's theorem \citep{Nelsen2006}, we can represent the joint density $G()$ using a copula $C(F_{X}(x), F_{D}(d), \theta)$. Given a fitted copula with dependence parameter $\hat{\theta}$, we can compute expected degree given a treatment status
\begin{align*}
E(d_{i}|x, \hat{\theta}) &= \int_{0}^{1} F^{-1}_{D}(p(u_{d} < U_{d}|F_{X}(x))) dU_{d},\\
&= \int_{0}^{1} F^{-1}_{D}( \frac{\partial C(u_{x}, u_{d}; \hat{\theta})}{\partial u_{x}}|_{u_{x} = F_{X}(x)}) dU_{d}.
\end{align*}

Therefore, the researcher can compute, for each $i$

\begin{equation*}
E(\sum_{j}b_{ij}x_{j}|x_{j}) = \sum_{j}E(b_{ij}|x_{j}, \hat{\theta})x_{j},
\end{equation*}

allowing the researcher to compute $\hat{\eta}(\hat{\theta})$.

This motivates a two-step estimator. In the first stage, the researcher estimates the copula from the empirical distribution of network statistics on a set of $M \leq N$ observations by picking the dependence parameter $\hat{\theta}$ that sets the score equal to zero

\begin{equation*}
\frac{1}{M}\sum_{i=1}^{M} \frac{\partial \ln C_{i}(F^{-1}_{x}, F^{-1}_{G}, \theta) }{\partial \theta} = 0.
\end{equation*}

 Given a value $\hat{\theta}$, the researcher then estimates the unobserved spillovers from the copula

\begin{equation*}
\hat{\eta}(\hat{\theta}) = \frac{\frac{1}{N} \sum_{i}(\sum_{j}h_{ij}x_{j})(\hat{\sum_{j}b_{ij}x_{j}}(\theta))}{\frac{1}{N}\sum_{i}(\sum_{j}h_{ij}x_{j})^{2}}.
\end{equation*}

and then constructs bias-corrected estimates as

\begin{equation*}
\hat{\beta} = \frac{\hat{\beta}^{\text{OLS}}}{1+\hat{\eta}(\hat{\theta})}.
\end{equation*}

This estimator is also consistent and asymptotically normal \citep{Newey1984, Smith2003}.\footnote{This result can be complicated when the researcher also has to estimate the underlying distributions that go into the copula. For a discussion of estimation issues, see \cite{Choros2010}.}

\begin{proposition}
    Make Assumptions 1A ,2,3, 5, A1. Define 

\begin{align*}
\begin{pmatrix}
h_{1}(\theta)\\
h_{2}(\theta, \beta)
\end{pmatrix} &=
\begin{pmatrix}
\frac{1}{M}\sum_{i=1}^{M} \frac{\partial \ln C_{i}(F^{-1}_{x}, F^{-1}_{G}, \theta) }{\partial \theta}.\\
\frac{1}{N}\sum_{i}(\sum_{j}h_{ij}x_{j})(y_{i} - (1+\eta(\theta))\beta(\sum_{j}h_{ij}x_{j}))
\end{pmatrix}.\\
\begin{pmatrix}
K_{11} & K_{12}\\
K_{21} & K_{22}
\end{pmatrix}
&= \plim \frac{1}{N} \sum_{i} E
\begin{pmatrix}
\frac{\partial \ln C_{i}(F^{-1}_{x}, F^{-1}_{G}, \theta) }{\partial \theta} & 0\\
-  \frac{\partial \eta(\theta)}{\partial \theta}\beta (\sum_{j}h_{ij}x_{j})^{2}  & -(1+\eta(\theta))(\sum_{j}h_{ij}x_{j})^{2}  
\end{pmatrix}\\
\begin{pmatrix}
S_{11} & S_{12}\\
S_{21} & S_{22}
\end{pmatrix} &=
\plim \frac{1}{N} \sum_{i} E \begin{pmatrix}
h_{1i}h_{1i}'& h_{2i}h_{1i}'\\
h_{2i}h_{1i}' & h_{2i}h_{2i}'
\end{pmatrix}
\end{align*}

$\hat{\beta}$ is a consistent estimator of $\beta$. Furthermore,

    \begin{equation*}
\sqrt{N}(\hat{\beta} - \beta) \sim N(0, K_{22}^{-1}(S_{22} + K_{21}K_{11}^{-1}S_{11}K_{11}^{-1}K_{21}' - K_{21} K^{-1}_{11}S_{12}  -S_{21}K^{-1}_{11}K_{21}') K_{22}^{-1}).
\end{equation*}

\end{proposition}
 
 Of course, implementing this estimator requires that the researcher can fit the copula and construct a link between the sampled degree statistic and the unobserved link weights. How the researcher might implement this depends on the sampling rule. For example, if the network is sampled using a fixed choice design, there exist some (low degree) nodes where treatment status and in-degree are fully observed. The researcher can fit the copula on this subset of individuals, under the assumption that the dependence between treatment and degree is the same for low and high degree nodes.\footnote{Then, the researcher may use a truncation invariant copula. See \cite{Nelsen2006} for further discussion.} The dependence parameter of the copula is an aggregate statistic. So a data provider could disclose this from a full data source without violating individual privacy. If the researcher is unsure how the degree statistic maps onto the unsampled link weights, they could use the result to construct bounds on the estimated spillover effect by sampling $\eta$ under different assumptions.

 \textbf{Example -- classroom setting, fixed choice sampling rule.} Suppose that children are assigned a continuous treatment, that we can describe with a marginal distribution $x_{i} \sim N(5,1)$. In an effort to maximise the effect of the educational intervention, school administrators have given larger doses to children the more friends that they have. Therefore, treatment status depends upon the network structure -- the child's degree. The researcher samples friendships using a fixed choice sampling design (asking them to name up to $m$ friends) with $m=5$. 

 The researcher observes the treatment status of each child, and their sampled number of friends. To construct $\hat{\eta}$, the researcher needs to estimate the number of expected treated friends for each child with five sampled friends by fitting the copula. In the first step, the researcher can fit the copula on the subsample of children whose degree and treatment are fully observed. There are the children with fewer than five friends. Specifically, the researcher can model the marginal distributions of child's treatments and degrees as coupled through a bivariate Gumbel copula 
\begin{equation*}
C(F_{X}^{-1}(x), F_{D}^{-1}(d); \theta) = \text{exp}(-((-\ln{F_{X}^{-1}(x)})^{\theta} + (- \ln{F_{D}^{-1}(d)})^{\theta})^{\frac{1}{\theta}})
\end{equation*}

where $\theta \in [1, \infty]$ controls the degree of dependence between treatment and degree. Fitting the copula gives an estimate of the dependence parameter $\hat{\theta}$. Then, we can use the fitted copula to estimate $\hat{\eta}$ using 

\begin{equation*}
\sum_{j}E(b_{ij}(x_{i}, \hat{\theta})|x_{j}) x_{j} = \sum_{j} (E(g_{ij}^{*}|x_{i}, \hat{\theta}) - m)\bar{x}.
\end{equation*}

 To assess how well this strategy performs in finite sample, we provide simulation results for this case in Appendix A.7.

%Here, we briefly discuss how researchers can construct bias-corrected estimators in these cases by modelling the dependence between treatment status and network statistics through a copula \citep{Nelsen2006, Smith2003}. The benefit of doing this as opposed using estimators that require the researcher to specify counterfactual shock exposure processes \citep[e.g ][]{Borusyak2023, Herstad2023} is that researchers can flexibly fit copulas from marginal distributions without having a specify a parametric model of the network formation and shock assignment process.

% But, assumption 1-B is not appropriate.  to obtain the expected dependence between observed and unobserved spillovers

%\begin{equation*}
%\frac{1}{N}\sum_{i} \Big(p(i \in \mathcal{B})E \Big((\sum_{j}h_{ij}x_{j})(\sum_{j}b_{ij}x_{j})|i \in \mathcal{B}\Big)\Big).
%\end{equation*}

%%%%%%%%%%%%%%%%%%%%%%%%%%%%%%%%%%%%%%%%%%%%%%%%%%%%%%%%%
\section{Simulation experiments}
\label{sec:simulations}
%%%%%%%%%%%%%%%%%%%%%%%%%%%%%%%%%%%%%%%%%%%%%%%%%%%%%%%%%
\FloatBarrier

Next, we evaluate the bias introduced by common sampling schemes and the performance of our rescaled estimators by Monte-Carlo simulation. Standard regression estimators can be heavily biased. The size of the bias depends on how much the sampling scheme alters the true network. Bias-corrected estimators perform well in finite samples. The distribution of bias-corrected estimates using $\hat{\eta}$ is close to the distribution of bias-corrected estimates under the true $\eta$, which is unbiased. In the appendix, we also simulate the performance of our estimator on real-life economic network that has been completely sampled -- the co-authorship network of economists in \cite{Ductor2014}.

\subsection{Simulated networks}

In each simulation, there are $N=1000$ individuals who receive a binary treatment $x_{i} \sim \text{Bernoulli}(0.3)$. In each case, outcomes are drawn from (\ref{eq:structural_model})
with $\beta = 0.8$, $\epsilon_{i} \sim N(0,1)$. We consider five networks and sampling schemes.
\begin{enumerate}
    \item \textbf{Fixed choice design}. Each individual draws an in-degree from a discrete uniform distribution $d_{i}\sim U(1,15)$.\footnote{We use a uniform distribution and sample neighbours uniformly at random from the population here to emphasise that the size of the bias that we find is not driven by tail behaviour of the degree distribution or preferential attachment-type mechanisms. Similar results hold when node degrees are sampled from more natural degree distributions like a discrete Pareto distribution \citep{Clauset2009}.} We form a binary directed simple network by connecting each individual with others uniformly at random from the population. We then sample links coming into each individual using a fixed choice design with reporting thresholds $m \in {1, ..., 14}$.
    
    \item \textbf{Sampling based on groups}.  Each individual belongs to a single group (e.g high school class). There are $20$ groups of $25$ individuals, $10$ groups of $20$ individuals, and $20$ groups of $15$ individuals. The researcher samples each individual as linked to every other individual in their group. True degrees are drawn $U(m_{i}-k, m_{i}-5-k)$, where $m_{i}$ is their group size and $k \in \{1,2,3,4,5\}$.
    
    \item \textbf{Link weight thresholds}. Each individual draws interaction intensities with others from $w_{ij}\sim \text{LogNormal}(1,15)$.\footnote{The exact setting is calibrated similarly to the model of the US public-firm production network in \cite{Herskovic2020}. } Then, we construct a weighted network where $g_{ij} = \frac{w_{ij}}{\sum_{k}w_{ik}}$. We sample links where $g_{ij}$ exceeds a threshold $ \tau \in \{0.025, 0.05, ..., 0.2\}$.
    
    \item \textbf{Fixed choice design with weights}. Each individual draws an in-degree from a discrete uniform distribution $d_{i}\sim U(1,15)$. We form a weighted directed simple network connecting individuals with others uniformly at random. Weights are $g_{ij} = \frac{1}{d_{i}}$ -- individuals who have more friends allocate less weight to each friend. Therefore reported weighted degree depends on number of friends. We then sample weighted links coming into each individual using a fixed choice design with reporting thresholds $m \in {1, ..., 14}$.

    \item \textbf{Sampling based on groups, true degree depends on group size}. Each individual belongs to one group (e.g high school class). There are $20$ groups of $25$ individuals, $10$ groups of $20$ individuals, and $20$ groups of $15$ individuals. The researcher samples each individual as linked to every other individual in their group. Their degrees are drawn $U(25-3 k, 20-3k)$, $U(20-2k, 15-2k)$, $U(15-k, 10-k)$ for each group respectively, where $k \in \{1,2,3,4,5\}$. 

\end{enumerate}

In the first three cases, assumption \ref{asm:indep_sampled_unsampled} holds. In the final two cases, only assumption \ref{asm:cond_sampled_unsampled} holds. In each case, we construct estimates of $\beta$

\begin{enumerate}
\item by regressing outcomes on spillovers on the sampled network (\ref{eq:naive_ols}),
\item using the bias-corrected estimator given the true $\eta$ (\ref{eq:estimator_true_eta}), and 
\item using the bias-corrected estimator estimating $\hat{\eta}$ using the results from section \ref{sec:rescale}.
\end{enumerate}
 We run $1000$ simulations per estimator, and report average values across each simulation. Additional experiments, including simulations for non-linear models and using the co-authorship network of economists, are given in Appendix A.8.

 Below, we compare mean estimates across simulations, and plot the distribution of a representative set of estimates from each setting.

\begin{figure}[!htbp]
 \begin{minipage}{\textwidth}
  \begin{minipage}[b]{0.475\textwidth}
    \centering
\input{mc_table_1}
\label{tab:fig1}
      \captionsetup{type=table}
    \captionof{table}{Mean spillover estimates using fixed choice design, by threshold}
   \end{minipage} \hspace{0.05\textwidth}%
    \begin{minipage}[b]{0.475\textwidth}
\includegraphics[width=0.95\textwidth]{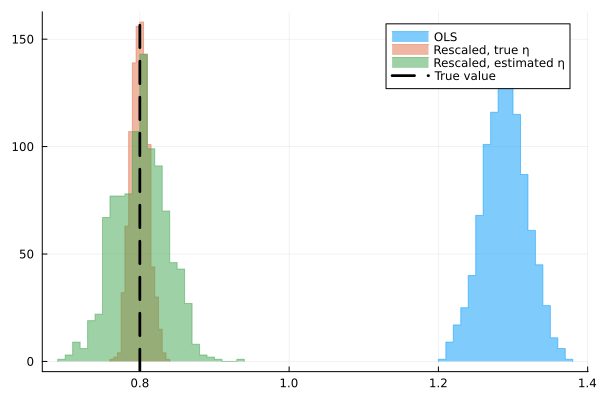}
\captionof{figure}{Distribution of spillover estimates using fixed choice design with threshold of $5$}
    \end{minipage}
\end{minipage}

\vspace{1em} 

 \begin{minipage}{\textwidth}
  \begin{minipage}[b]{0.475\textwidth}
    \centering
\input{mc_table_2}
\label{tab:fig1} 
\captionsetup{type=table}
\captionof{table}{Mean spillover estimates sampling based on groups, by $K$}

   \end{minipage} \hspace{0.05\textwidth}
    \begin{minipage}[b]{0.475\textwidth}
\includegraphics[width=0.95\textwidth]{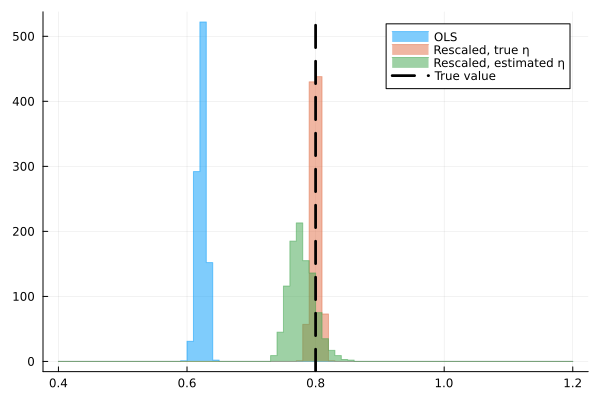}
\caption{Distribution of spillover estimates sampling based on groups, $k=3$}
    \end{minipage}
\end{minipage}

\vspace{1em}

 \begin{minipage}{\textwidth}
  \begin{minipage}[b]{0.475\textwidth}
    \centering
\input{mc_table_3}
\label{tab:fig1} 
\captionsetup{type=table}
\captionof{table}{Spillover estimates using fixed choice design, by threshold}

   \end{minipage} \hspace{0.05\textwidth}
    \begin{minipage}[b]{0.475\textwidth}
\includegraphics[width=0.95\textwidth]{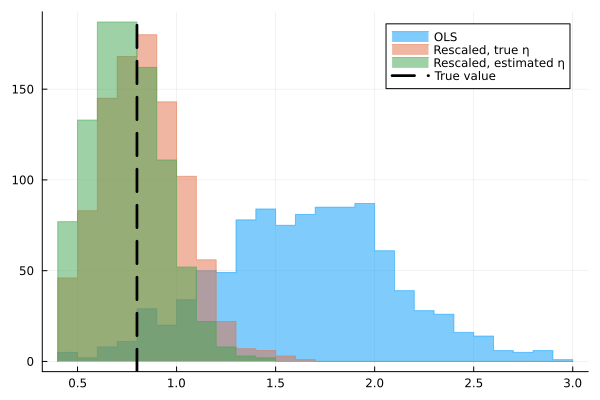}
\caption{Distribution of spillover estimates using fixed choice design with threshold of $5$}
    \end{minipage}
\end{minipage}
\end{figure}

\begin{figure}[!htbp]
 \begin{minipage}{\textwidth}
  \begin{minipage}[b]{0.475\textwidth}
    \centering
\input{mc_table_4}
\label{tab:fig1} 
\captionsetup{type=table}
\captionof{table}{Mean spillover estimates from a fixed choice design with weights, by number sampled}

   \end{minipage}\hspace{0.05\textwidth}
    \begin{minipage}[b]{0.475\textwidth}
\includegraphics[width=0.95\textwidth]{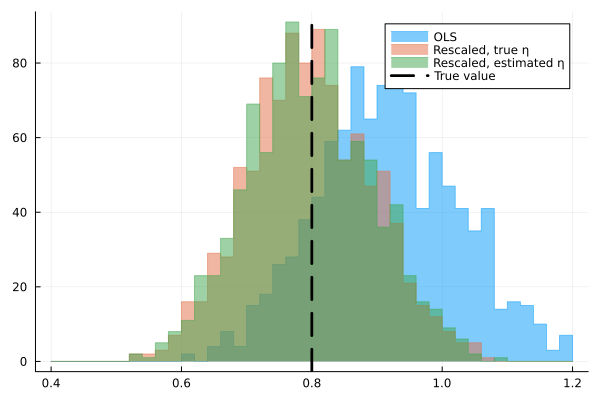}
\caption{Distribution of spillover estimates using fixed choice design with threshold of $5$}
    \end{minipage}
\end{minipage}

\vspace{1em}

 \begin{minipage}{\textwidth}
  \begin{minipage}[b]{0.475\textwidth}
    \centering
\input{mc_table_5}
\label{tab:fig1} 
\captionsetup{type=table}
\captionof{table}{Mean spillover estimates sampling based on groups when true degree depends on group size, by $k$}

   \end{minipage}\hspace{0.05\textwidth}
    \begin{minipage}[b]{0.475\textwidth}
\includegraphics[width=0.95\textwidth]{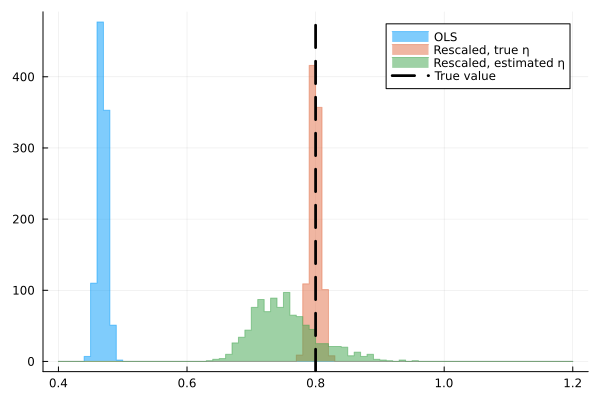}
\caption{Distribution of spillover estimates using fixed choice design with threshold of $5$}
    \end{minipage}
\end{minipage}

\end{figure}

Simply regressing outcomes on sampled spillovers yields biased estimates. As expected, estimates are too large when we sample a subset of the true links between individuals (cases 1, 3, and 4) and too small when we sample a superset (cases 2 and 5). Bias can be substantial. For example, in the case of a fixed choice design sampling at most five links per individual \citep[as for within-gender friendships in the Ad-Health dataset][]{Harris2009}, the average spillover effect estimates is $1.28$ -- $1.6$ times the true effect. Thresholding links based at $10 \%$ of total flows \citep[similar to how supply links are sampled between U.S. public firms][]{Atalay2011}, gives average spillover effect estimates of $1.63$ -- double the true effect.   

Our bias-corrected estimators recover the true spillover effect well in finite samples. With $\eta$ known, estimators are almost always centered on the true spillover value. With $\hat{\eta}$ estimated under assumption 4.a or assumption 4.b, estimates are centered very close to the true value. Bias-corrected estimators perform well when either assumptions hold, particularly under fixed choice sampling designs (cases 1 and 4). 

%%%%%%%%%%%%%%%%%%%%%%%%%%%%%%%%%%%%%%%%%%%%%%%%%%%%%%%%%
\section{Propagation of climate shocks in production networks}
\label{sec:empirical_application}
%%%%%%%%%%%%%%%%%%%%%%%%%%%%%%%%%%%%%%%%%%%%%%%%%%%%%%%%%

As an example, we use our estimator to measure how climate shocks propagate across supply links between public firms in the United States using self-reported supply relationships. Appendix A.9, also considers peer effects in education \citep{Carrell2013}.

There is a consensus that a central effect of climate change is an increase in extreme weather events \citep[e.g see][ and references therein]{Robinson2021}. Whether these types of idiosyncratic shocks propagate between firms matters for the effect of climate change on economic output \citep{Barrot2016}. If firms can easily substitute away from suppliers hit by extreme weather shocks, the impact remains limited to those suppliers. If, however, shocks propagates from suppliers to customers, supply chains amplify the direct effect of these shocks \cite[e.g see][]{Carvalho2020}.

\subsection{Balance-sheet and supply-chain data}

Data on supply links between U.S. public firms come from the popular Compustat Supply Chain dataset \citep{Atalay2011}. Since 1997, SFAS regulation No. 131 has required public firms to report customers that account for more than 10\% of sales in 10-K filings with the Securities and Exchange Commission. Firms may report other customers voluntarily. Compustat collects all of these self-reported links, which are understood to be a subset of the true supply links \citep{Herskovic2020, Bacilieri2023}.\footnote{Before the introduction of the regulation in $1997$, firms would self-report certain customers. Some firms also report additional customers. For more details, see \cite{Bacilieri2023}.} In $2017$ , the mean number of reported suppliers in is $1.36$, and the median is $0.00$. This is far fewer than researchers see in complete transactions data, and is an example where researchers can only sample high-weight links.\footnote{For example, the mean number of suppliers in Belgian production network data is $\approx 30$ \citep{Dhyne2021}, in Chilean data is $\approx 20$ \citep{Huneeus2020}, and in Ecuadorian data is $\approx 33$ \citep{Bacilieri2023}. The degree distribution is shifted to the left compared to true networks from VAT data, that shows similar patterns across countries \citep{Bacilieri2023}. Furthermore, \cite{Bacilieri2023} analyse a larger sample of self-reported network from 2012-2013, and find that 27 percent of firms have no listed suppliers, and 30 percent have no listed customers. The high amount of isolated firms suggests that some paths between firms are missing entirely.} As described earlier, this sampling scheme satisfies assumption \ref{asm:indep_sampled_unsampled}.\footnote{As in \cite{Barrot2016}, we treat the underlying network as binary. Further research could account for the effect of weights.} So, we can construct rescaled estimates based on the mean missing degree amongst firms with at least one missing link. We use the mean degree of the (more complete) Factset production network in \citep{Bacilieri2023}, and value accounting for truncation in \cite{Herskovic2020}. The resulting values are $2.7, 2.56$, corresponding to a mean missing degree of $d^{B} = 1.2, 1.34$.

Firm-level balance-sheet information for $1711$ U.S. public firms in $2017$ comes from the Compustat Fundamentals Quarterly North America dataset. Continuous variables are winsorized at the 99th and 1st percentiles. As firms may relocate headquarters, we locate firms using addresses reported in their 10-K forms instead of the location reported in Compustat \citep{Gao2021}.

%Constructing the mean missing degree has to be done carefully, as the mean degree of production networks is observed to scales with the number of nodes \cite{Bacilieri2023}. So, we construct the mean missing degree we using the shape of the true degree distribution of similar binary firm-level production networks from \citep{Herskovic2020, Bacilieri2023}. 

%As observed in complete production network datasets, we describe the true degree distribution by a discrete power law distribution that is top-censored at $N-1$ \citep{Bacilieri2023}. We take the estimated tail exponents from the Factset dataset -- a more completely sampled dataset of similar types of firms to the US public firms in \cite{Barrot2016}'s sample -- and from \cite{Herskovic2020}'s study of the same US public firms. Finally, we need to compute the mean degree for a network with the reported number of firms, which we do by Monte-Carlo sampling using the discrete power-law sampler from \cite{Clauset2009}. We get values of $2.7, 2.56$. With this, we get estimates of the mean missing degree of $d^{B} = 1.2, 1.34$.
 
\subsection{Climate shocks}

\begin{sidewaystable}
   \caption{Major climate disasters in the United States, $2017$} 
    \centering
    \input{weather_shock_tab}
 \begin{flushleft}
 \footnotesize \textbf{Notes: } Events come from the NOAA Billion Dollar Weather and Climate Disasters Project. Affected states are those in which at least one county declares a state of emergency associated with the disaster as listed in the FEMA Disaster Declarations Dataset. Events that last longer than one month, or where no county declared a state of emergency, are excluded. 
\end{flushleft}
\label{weather_shock_table}
\end{sidewaystable}

To determine which firms receive weather shocks, we construct a dataset of the county-level incidence of severe weather events in the United States $2004$-$2019$.\footnote{The dataset is available on request.} Events comes from the US National Oceanic and Atmospheric Administration Billion-Dollar Weather and Climate Disasters project, which lists all weather events causing over $\$1$ billion in damages ($2024$ dollars) between $1980$ and $2024$.\footnote{See \url{https://www.ncei.noaa.gov/access/billions}} We match each weather event to county-level emergency declarations from the Federal Emergency Management Agency.

A county is coded as affected by a disaster if is in a state affected by the disaster and they have declared a state of emergency from that type of natural disaster (e.g a flood, a storm) in the days around the event given by the US National Oceanic and Atmospheric Administration. This yields a dataset of each county affected by a `billion-dollar' natural disaster by month. 

Table 1. lists the extreme weather events in the United States in 2017. There are six disasters in our dataset within this year: three hurricanes, two outbreaks of tornadoes, and one case of significant flooding. They affect firms in nine states over five months of the year. Total estimated damages range between $ \$1.2$ billion and $\$160$ billion per disaster.

Following \cite{Barrot2016}, a firm is classified as hit by a shock in a given quarter if it is headquartered in a county affected by the disaster in that quarter. $14.9 \%$ of firms are hit with at least one weather shock within the year, and $11.3 \%$ have at least one reported supplier affected. There is strong evidence that firms do not choose suppliers based on the distribution of weather shocks across space \citep{Barrot2016}. So we can treat the distribution of these shocks as independent of the distribution of supply links between firms.

\FloatBarrier

\subsection{Estimation}

We estimate the effect of an additional shock to a firm's supplier over a year on that year's sales growth, accounting for shocks to the firm itself, using the regression model

\begin{equation*}
\Delta \ln{\text{Sales}}_{it, t-4} = \alpha + \beta_{1} \sum_{j}h_{ij}\text{Shocked}_{jt, t-4} + \beta_{2}\text{Shocked}_{t,t-4} + X \gamma+  \epsilon_{i}. 
\end{equation*}

We construct bias-corrected estimates of $\beta_{1}$ using the estimator (\ref{eq:rescaled_estimator_1}) under the assumption that missingness is independent of controls. 

 \begin{table}[!h]
   \caption{Estimates of propagation of climate shocks between US public firms over $2017$} 
 \centering
\input{reg_tab}
 \begin{flushleft}
 \footnotesize \textbf{Notes:} Standard errors for non-rescaled estimates clustered by county (the level of shock assignment). Standard errors for rescaled estimates bootstrapped with 10000 draws. Firm-level controls are size (ppentq) and industry (4-digit NAICS fixed effects).
\end{flushleft}
\label{tab:shock_table}
 \end{table}

Table \ref{tab:shock_table} reports results. In line with the existing literature \citep{Barrot2016}, the uncorrected estimate suggests that a shock to an additional supplier within the year leads to a $2.48 \%$ fall in yearly sales growth.
After bias-correction, spillover effects are $53$-$56 \%$ of the initial estimates. Almost half of the naive spillover effect appears to come from bias due to measurement error in links. We cannot reject the null hypothesis that spillover are zero at standard significance levels. Looking at robustness of the estimates to sampling bias using (\ref{eq:robustness_degree}) suggests that estimates are very sensitive to missing links. Estimates fall to less than a $1.5 \%$ percent drop in yearly sales growth if we are missing at least one link on average, and a $1 \%$ fall if we are missing at least $2.25$ links on average. Economically, these results may reflect the short duration of most weather shocks. Customers may be able to smooth out these short-term disruptions using inventories. We would expect the effects of these types of shocks to be smaller than those of larger natural disasters that cause long term disruption \citep[e.g][]{Carvalho2020}.

\FloatBarrier
%%%%%%%%%%%%%%%%%%%%%%%%%%%%%%%%%%%%%%%%%%%%%%%%%%%%%%%%%
\section{Conclusion}
%%%%%%%%%%%%%%%%%%%%%%%%%%%%%%%%%%%%%%%%%%%%%%%%%%%%%%%%%

We first show that sampling links between individuals can lead to substantial, economically significant bias in spillover estimates from linear and nonlinear models. Unlike classical measurement error, which generates downward bias, sampling can create either upward or downward bias depending on the scheme. Simulations demonstrate that popular sampling schemes lead to economically significant biases in estimates.

To solve this, we introduce bias-corrected estimators that rescale linear and nonlinear regressions to account for dependence between spillovers on observed and unobserved links. In experimental and quasi-experimental settings, researchers can implement these estimators using only aggregate statistics of the degree distribution, which are relatively easy to sample. Our estimators perform well in simulations. To illustrate, we estimate the propagation of climate shocks among US public firms in 2017 using sampled supply links.

For tractability, we rely on the linearity of the estimators in the sampled and unsampled networks for our results. Applied economists commonly fit complex structural models to sampled network data where parameters are non-linear functions of sampled networks \citep[e.g see]{Badev2021, Lim2024}. Future work could extend our results to moment-based estimators in these settings. Our findings underscore that careful treatment of network sampling is essential for credible empirical estimates of spillovers.

\newpage

\bibliographystyle{apalike}
\bibliography{references}
\clearpage
\FloatBarrier

%\section{Revisiting the LIN PAPER}
%%%%%%%%%%%%%%%%%%%%%%%%%%%%%%%%%%%%%%%%%%%%%%%%%%%%%%%%%
\section*{Appendix}
 \renewcommand{\theequation}{A-\arabic{equation}}
  \renewcommand{\thesection}{A\arabic{section}}
\setcounter{section}{0}
\setcounter{figure}{1}
\counterwithin{figure}{section}
\label{sec:appendix}
%%%%%%%%%%%%%%%%%%%%%%%%%%%%%%%%%%%%%%%%%%%%%%%%%%%%%%%%%
\FloatBarrier

%%%%%%%%%%%%%%%%%%%%%%%%%%%%%%%%%%%%%%%%%%%%%%

\section{Proofs}

We make the following standard assumptions for asymptotic results \citep{Cameron2005}
\begin{assumption}
\label{asm:concentration_limit}
The matrix with entries $\plim \frac{1}{N}\sum_{i}\epsilon^{2}_{i}\sum_{j}g_{ij}x_{j} \sum_{j}g_{kj}x_{j}$ exists and is finite positive definite. Furthermore
\begin{align*}
\plim \frac{1}{N}\sum_{i} (\sum_{j}h_{ij}x_{j})^{2} &= E((\sum_{j}h_{ij}x_{j})^{2})\\
\plim \frac{1}{N}\sum_{i} (\sum_{j}h_{ij}x_{j})(\sum_{j}b_{ij}x_{j}) &= E((\sum_{j}h_{ij}x_{j})(\sum_{j}b_{ij}x_{j})),\\
\exists \delta > 0 \text{ s.t } E(|\sum_{j}g_{ij}x_{j} \sum_{j}g_{kj}x_{j}|^{1+\delta}) &\leq \infty \text{ }\forall k, i\\
\exists \delta > 0 \text{ s.t } E(|\epsilon^{2}_{i}|^{1+\delta}) &\leq \infty \text{ }\forall k, i\\
\exists \delta > 0 \text{ s.t } E(|\epsilon^{2}_{i}\sum_{j}g_{ij}x_{j} \sum_{j}g_{kj}x_{j}|^{1+\delta}) &\leq \infty \text{ }\forall k, i
\end{align*}
\end{assumption}

Note that these may fail if the network has a degree distribution that is heavy tailed \citep[see][]{Newman2010}. Examples include a power-law degree distribution. We do not address this, as this is separate to the focus of this paper. In this setting, estimation of spillover effects by regression models would need further justification in general.

\paragraph{Proof of proposition 1}

\begin{proof}

\begin{equation}
\label{eq:naive_ols}
\hat{\beta}^{\text{OLS}} = \beta \Big(1+\frac{\frac{1}{N}\sum_{i} (\sum_{j}h_{ij}x_{j})(\sum_{j}b_{ij}x_{j})}{\frac{1}{N}\sum_{i} (\sum_{j}h_{ij}x_{j})^{2}} \Big) + \frac{\frac{1}{N}\sum_{i} (\sum_{j}h_{ij}x_{j})\epsilon_{i}}{\frac{1}{N}\sum_{i} (\sum_{j}h_{ij}x_{j})^{2}}.
\end{equation}

Therefore,

\begin{align*}
E(\hat{\beta}^{\text{OLS}}) &= \beta E\Big(1+\frac{\frac{1}{N}\sum_{i} (\sum_{j}h_{ij}x_{j})(\sum_{j}b_{ij}x_{j})}{\frac{1}{N}\sum_{i} (\sum_{j}h_{ij}x_{j})^{2}} \Big) + E\Big(\frac{\frac{1}{N}\sum_{i} (\sum_{j}h_{ij}x_{j})\epsilon_{i}}{\frac{1}{N}\sum_{i} (\sum_{j}h_{ij}x_{j})^{2}}\Big).\\
&= \beta  + \beta E\Big(\frac{\frac{1}{N}\sum_{i} (\sum_{j}h_{ij}x_{j})(\sum_{j}b_{ij}x_{j})}{\frac{1}{N}\sum_{i} (\sum_{j}h_{ij}x_{j})^{2}} \Big) + E\Big(\frac{\frac{1}{N}\sum_{i} (\sum_{j}h_{ij}x_{j})}{\frac{1}{N}\sum_{i} (\sum_{j}h_{ij}x_{j})^{2}} E(\epsilon_{i}|\sum_{i}h_{ij}x_{j})\Big).
\end{align*}

Under assumption 2

\begin{align*}
E(\epsilon_{i}|\sum_{i}h_{ij}x_{j}) &= E(\epsilon_{i})\\
&= 0.
\end{align*}

By assumption, 

\begin{equation*}
E\Big(\frac{\frac{1}{N}\sum_{i} (\sum_{j}h_{ij}x_{j})(\sum_{j}b_{ij}x_{j})}{\frac{1}{N}\sum_{i} (\sum_{j}h_{ij}x_{j})^{2}} \Big) \neq 0.
\end{equation*}

The proposition follows.

\end{proof}

\paragraph{Proof of proposition 2}

\begin{proof}
\begin{align*}
E \Big(\frac{1}{N}\sum_{i} (\sum_{j}h_{ij}x_{j})(\sum_{j}b_{ij}x_{j}) \Big) &= \frac{1}{N}\sum_{i}E \Big((\sum_{j}h_{ij}x_{j})(\sum_{j}b_{ij}x_{j}) \Big) \text{ by linearity of $E()$},\\
&= \frac{1}{N}\sum_{i} \Big( p(i \notin \mathcal{B})E \Big((\sum_{j}h_{ij}x_{j})(\sum_{j}b_{ij}x_{j})|i \notin \mathcal{B}\Big), \\
&+ p(i \in \mathcal{B})E \Big((\sum_{j}h_{ij}x_{j})(\sum_{j}b_{ij}x_{j})|i \in \mathcal{B}\Big)\Big) \text{splitting those with no incorrectly sampled links,}\\
&= \frac{1}{N}\sum_{i} \Big( 0 + p(i \in \mathcal{B})E \Big((\sum_{j}h_{ij}x_{j})(\sum_{j}b_{ij}x_{j})|i \in \mathcal{B}\Big)\Big)\\
&=  \frac{1}{N}\sum_{i}p(i \in \mathcal{B}) (E(x)^{2} E((\sum_{j}h_{ij})(\sum_{j}b_{ij})|i \in \mathcal{B}\Big)\Big) \text{under assumption 1,}\\\
&= \frac{1}{N}\sum_{i}p(i \in \mathcal{B}) (E(x)^{2} E((\sum_{j}h_{ij})E(\sum_{j}b_{ij}|\sum_{j}h_{ij})|i \in \mathcal{B}\Big)\Big) \text{conditioning,}\\
&= \frac{1}{N}\sum_{i}p(i \in \mathcal{B}) (E(x)^{2} E(d^{H}_{i}E(d^B_{i}|d^{H}_{i})|i \in \mathcal{B}\Big)\Big).
\end{align*}

We look for the cases when this term is non-zero. Assume that $E(x) \neq 0, \text{ and } p(i \in \mathcal{B}) \neq 0$. Then, it is equivalent to

\begin{equation*}
\sum_{i} E \Big( d^{H}_{i}E(d^B_{i}|d^{H}_{i})|i \in \mathcal{B}\Big) \neq 0.
\end{equation*}

Assume that $d^{H}_{i}$ has the same sign for each $i$. Then a sufficient condition for this to be non-zero is that $E(d^B_{i}|d^{H}_{i})$ is either non-negative or non-positive for each $i$ such that $i \in \mathcal{B}$.

\end{proof}

\paragraph{Proof of theorem 1}

\begin{proof}
\begin{align*}
E(\hat{\beta}) &= E(\frac{\hat{\beta}^{\text{OLS}}}{1+ \eta})\\
&= \frac{1}{1+E \Big(\frac{\frac{1}{N}\sum_{i} (\sum_{j}h_{ij}x_{j})(\sum_{j}b_{ij}x_{j})}{\frac{1}{N}\sum_{i} (\sum_{j}h_{ij}x_{j})^{2}} \Big)} E(\hat{\beta}^{\text{OLS}})\\
&= \frac{1}{1+E \Big(\frac{\frac{1}{N}\sum_{i} (\sum_{j}h_{ij}x_{j})(\sum_{j}b_{ij}x_{j})}{\frac{1}{N}\sum_{i} (\sum_{j}h_{ij}x_{j})^{2}} \Big)} E\Big(\beta \Big(1+ E\Big( \frac{\frac{1}{N}\sum_{i} (\sum_{j}h_{ij}x_{j})(\sum_{j}b_{ij}x_{j})}{\frac{1}{N}\sum_{i} (\sum_{j}h_{ij}x_{j})^{2}}\Big) \Big) + \frac{\frac{1}{N}\sum_{i} (\sum_{j}h_{ij}x_{j})\epsilon_{i}}{\frac{1}{N}\sum_{i} (\sum_{j}h_{ij}x_{j})^{2}} \Big) \text{(prop 1)}, \\
&= \beta + E \Big(\frac{\frac{1}{N}\sum_{i} (\sum_{j}h_{ij}x_{j})\epsilon_{i}}{(\frac{1}{N}\sum_{i} (\sum_{j}h_{ij}x_{j})^{2})(1+E \Big(\frac{\frac{1}{N}\sum_{i} (\sum_{j}h_{ij}x_{j})(\sum_{j}b_{ij}x_{j})}{\frac{1}{N}\sum_{i} (\sum_{j}h_{ij}x_{j})^{2}} \Big))} \Big)\\
&= \beta + 0 \text{ from assumption 2}.
\end{align*}
\end{proof}

Now, we prove consistency.
 \begin{proof}
 Our estimator is

\begin{equation*}
\hat{\beta} = \frac{1}{1+E\Big(\frac{\frac{1}{N}\sum_{i} (\sum_{j}h_{ij}x_{j})(\sum_{j}b_{ij}x_{j})}{\frac{1}{N}\sum_{i} (\sum_{j}h_{ij}x_{j})^{2}} \Big)} \Big(\beta \Big(1+ \Big( \frac{\frac{1}{N}\sum_{i} (\sum_{j}h_{ij}x_{j})(\sum_{j}b_{ij}x_{j})}{\frac{1}{N}\sum_{i} (\sum_{j}h_{ij}x_{j})^{2}}\Big) \Big) + \frac{\frac{1}{N}\sum_{i} (\sum_{j}h_{ij}x_{j})\epsilon_{i}}{\frac{1}{N}\sum_{i} (\sum_{j}h_{ij}x_{j})^{2}} \Big). 
\end{equation*}

First, consider the term

\begin{align*}
& \plim \frac{1}{1+E\Big(\frac{\frac{1}{N}\sum_{i} (\sum_{j}h_{ij}x_{j})(\sum_{j}b_{ij}x_{j})}{\frac{1}{N}\sum_{i} (\sum_{j}h_{ij}x_{j})^{2}} \Big)} \beta \Big(1+ \Big( \frac{\frac{1}{N}\sum_{i} (\sum_{j}h_{ij}x_{j})(\sum_{j}b_{ij}x_{j})}{\frac{1}{N}\sum_{i} (\sum_{j}h_{ij}x_{j})^{2}}\Big) \Big)\\
\end{align*}

Then, applying Slutsky's lemma

\begin{align*}
\plim\beta \Big(1+ \Big( \frac{\frac{1}{N}\sum_{i} (\sum_{j}h_{ij}x_{j})(\sum_{j}b_{ij}x_{j})}{\frac{1}{N}\sum_{i} (\sum_{j}h_{ij}x_{j})^{2}}\Big) \Big) &= \beta + \beta \frac{E((\sum_{j}h_{ij}x_{j})(\sum_{j}b_{ij}x_{j}))}{E((\sum_{j}h_{ij}x_{j})^{2})}.
\end{align*}

Consider the Taylor expansion of $E \Big(\frac{\frac{1}{N}\sum_{i} (\sum_{j}h_{ij}x_{j})(\sum_{j}b_{ij}x_{j})}{\frac{1}{N}\sum_{i} (\sum_{j}h_{ij}x_{j})^{2}} \Big)$ around $E((\sum_{j}h_{ij}x_{j})^{2}), E((\sum_{j}h_{ij}x_{j})(\sum_{j}b_{ij}x_{j}))$

\begin{align*}
 E \Big(\frac{\frac{1}{N}\sum_{i} (\sum_{j}h_{ij}x_{j})(\sum_{j}b_{ij}x_{j})}{\frac{1}{N}\sum_{i} (\sum_{j}h_{ij}x_{j})^{2}} \Big) &= 
 \frac{E(\frac{1}{N}\sum_{i} (\sum_{j}h_{ij}x_{j})(\sum_{j}b_{ij}x_{j}))}{E(\frac{1}{N}\sum_{i} (\sum_{j}h_{ij}x_{j})^{2})}\\
 & - \frac{\operatorname{Cov}(\frac{1}{N}\sum_{i} (\sum_{j}h_{ij}x_{j})(\sum_{j}b_{ij}x_{j}), \frac{1}{N}\sum_{i} (\sum_{j}h_{ij}x_{j})^{2})}{(E(\frac{1}{N}\sum_{i} (\sum_{j}h_{ij}x_{j})^{2}))^{2}}\\
 &+ \frac{\operatorname{Var}(\frac{1}{N}\sum_{i} (\sum_{j}h_{ij}x_{j})^{2})}{E((\frac{1}{N}\sum_{i} (\sum_{j}h_{ij}x_{j})^{2})^{3})} + ...,\\
\end{align*}

From assumption 6, $\operatorname{Var}(\frac{1}{N}\sum_{i} (\sum_{j}h_{ij}x_{j})^{2}) \rightarrow 0,$ and $\operatorname{Cov}(\frac{1}{N}\sum_{i} (\sum_{j}h_{ij}x_{j})(\sum_{j}b_{ij}x_{j}), \frac{1}{N}\sum_{i} (\sum_{j}h_{ij}x_{j})^{2}) \rightarrow 0$. Therefore

 \begin{equation*}
\plim E \Big(\frac{\frac{1}{N}\sum_{i} (\sum_{j}h_{ij}x_{j})(\sum_{j}b_{ij}x_{j})}{\frac{1}{N}\sum_{i} (\sum_{j}h_{ij}x_{j})^{2}} \Big) = \frac{E((\sum_{j}h_{ij}x_{j})(\sum_{j}b_{ij}x_{j}))}{E((\sum_{j}h_{ij}x_{j})^{2})}.
 \end{equation*}

 Combining these results, we have that

 \begin{equation*}
\plim \frac{1}{1+E\Big(\frac{\frac{1}{N}\sum_{i} (\sum_{j}h_{ij}x_{j})(\sum_{j}b_{ij}x_{j})}{\frac{1}{N}\sum_{i} (\sum_{j}h_{ij}x_{j})^{2}} \Big)} \beta \Big(1+ \Big( \frac{\frac{1}{N}\sum_{i} (\sum_{j}h_{ij}x_{j})(\sum_{j}b_{ij}x_{j})}{\frac{1}{N}\sum_{i} (\sum_{j}h_{ij}x_{j})^{2}}\Big) \Big) = \beta.
 \end{equation*}

 Next, consider the second term

 \begin{equation*}
\plim \frac{1}{1+E\Big(\frac{\frac{1}{N}\sum_{i} (\sum_{j}h_{ij}x_{j})(\sum_{j}b_{ij}x_{j})}{\frac{1}{N}\sum_{i} (\sum_{j}h_{ij}x_{j})^{2}} \Big)}\frac{\frac{1}{N}\sum_{i} (\sum_{j}h_{ij}x_{j})\epsilon_{i}}{\frac{1}{N}\sum_{i} (\sum_{j}h_{ij}x_{j})^{2}}.
 \end{equation*}

Under assumptions 1,2
\begin{align*}
\plim \frac{1}{N}\sum_{i} (\sum_{j}h_{ij}x_{j})\epsilon_{i} &=
0.
\end{align*}

Again applying Slutksy's lemma plus assumption A1 gives

 \begin{equation*}
\plim \frac{1}{1+E\Big(\frac{\frac{1}{N}\sum_{i} (\sum_{j}h_{ij}x_{j})(\sum_{j}b_{ij}x_{j})}{\frac{1}{N}\sum_{i} (\sum_{j}h_{ij}x_{j})^{2}} \Big)}\frac{\frac{1}{N}\sum_{i} (\sum_{j}h_{ij}x_{j})\epsilon_{i}}{\frac{1}{N}\sum_{i} (\sum_{j}h_{ij}x_{j})^{2}} = 0.
 \end{equation*}

 Combining our two intermediate results by Slutsky's lemma gives

 \begin{equation*}
\plim \hat{\beta} = \beta + 0.
 \end{equation*}
\end{proof}

\subsection{Proofs of proposition 3}

\begin{proof}
As before

\begin{equation*}
\eta = E \Big(\frac{\frac{1}{N}\sum_{i} (\sum_{j}h_{ij}x_{j})(\sum_{j}b_{ij}x_{j})}{\frac{1}{N}\sum_{i} (\sum_{j}h_{ij}x_{j})^{2}} \Big)
\end{equation*}

First, we want to show that we can approximate

\begin{equation*}
E \Big(\frac{\frac{1}{N}\sum_{i} (\sum_{j}h_{ij}x_{j})(\sum_{j}b_{ij}x_{j})}{\frac{1}{N}\sum_{i} (\sum_{j}h_{ij}x_{j})^{2}} \Big) \approx \frac{E(\frac{1}{N}\sum_{i} (\sum_{j}h_{ij}x_{j})(\sum_{j}b_{ij}x_{j}))}{E(\frac{1}{N}\sum_{i} (\sum_{j}h_{ij}x_{j})^{2})}.
\end{equation*}

From taking the Taylor expansion of this fraction around the point $\mu_{A}, \mu_{B}$, we can in general evaluate \citep{Billingsley2012}
\begin{equation*}
E \Big( \frac{A}{B} \Big) = \frac{\mu_{A}}{\mu_B} - \frac{\operatorname{Cov}(A,B)}{\mu_{B}^{2}} + \frac{\operatorname{Var}(B) \mu_{A}}{\mu^{3}_{B}} + \Delta.
\end{equation*}

Substituting 

\begin{align*}
A &= \frac{1}{N}\sum_{i} (\sum_{j}h_{ij}x_{j})(\sum_{j}b_{ij}x_{j}),\\
B &= \frac{1}{N}\sum_{i} (\sum_{j}h_{ij}x_{j})^{2},
\end{align*}

and solving gives 

\begin{align*}
 E \Big(\frac{\frac{1}{N}\sum_{i} (\sum_{j}h_{ij}x_{j})(\sum_{j}b_{ij}x_{j})}{\frac{1}{N}\sum_{i} (\sum_{j}h_{ij}x_{j})^{2}} \Big) &= 
 \frac{E(\frac{1}{N}\sum_{i} (\sum_{j}h_{ij}x_{j})(\sum_{j}b_{ij}x_{j}))}{E(\frac{1}{N}\sum_{i} (\sum_{j}h_{ij}x_{j})^{2})}\\
 & - \frac{\operatorname{Cov}(\frac{1}{N}\sum_{i} (\sum_{j}h_{ij}x_{j})(\sum_{j}b_{ij}x_{j}), \frac{1}{N}\sum_{i} (\sum_{j}h_{ij}x_{j})^{2})}{(E(\frac{1}{N}\sum_{i} (\sum_{j}h_{ij}x_{j})^{2}))^{2}}\\
 &+ \frac{\operatorname{Var}(\frac{1}{N}\sum_{i} (\sum_{j}h_{ij}x_{j})^{2})}{E((\frac{1}{N}\sum_{i} (\sum_{j}h_{ij}x_{j})^{2})^{3})} + ...,\\
 & =  \frac{E(\frac{1}{N}\sum_{i} (\sum_{j}h_{ij}x_{j})(\sum_{j}b_{ij}x_{j}))}{E(\frac{1}{N}\sum_{i} (\sum_{j}h_{ij}x_{j})^{2})} + \mathcal{O}\Big(\frac{1}{(\sum_{i}\sum_{j}h_{ij}x_{j})^{4}} \Big).
\end{align*}

where we disregard the final terms as they are vanishingly small. Next, we want to evaluate the top given that we do not observe $B$.

As in the proof of proposition 2, we can write

\begin{equation*}
E(\frac{1}{N}\sum_{i} (\sum_{j}h_{ij}x_{j})(\sum_{j}b_{ij}x_{j}))= \frac{1}{N}\sum_{i}p(i \in \mathcal{B}) (E(x)^{2} E(d^{H}_{i}E(d^B_{i}|d^{H}_{i})|i \in \mathcal{B}).
\end{equation*}

Now, applying assumption 4a, 

\begin{equation*}
E(d^{H}_{i}E(d^B_{i}|d^{H}_{i})|i \in \mathcal{B}) = E(d^{H}_{i}|i \in \mathcal{B}) E(d^{B}_{i}|i \in \mathcal{B})
\end{equation*}

Substituting back in, we have

\begin{equation*}
E(\frac{1}{N}\sum_{i} (\sum_{j}h_{ij}x_{j})(\sum_{j}b_{ij}x_{j}))= \frac{1}{N}\sum_{i}p(i \in \mathcal{B}) (E(x)^{2} E(d^{H}_{i}|i \in \mathcal{B}) E(d^{B}_{i}|i \in \mathcal{B}).
\end{equation*}

Substituting in the sample analogues and then applying Theorem 1 gives the results.

\end{proof}

\subsection{Proof of proposition 4}

\begin{proof}

From the proof of proposition 3,

\begin{equation*}
E(\frac{1}{N}\sum_{i} (\sum_{j}h_{ij}x_{j})(\sum_{j}b_{ij}x_{j}))= \frac{1}{N}\sum_{i}p(i \in \mathcal{B}) (E(x)^{2} E(d^{H}_{i}E(d^B_{i}|d^{H}_{i})|i \in \mathcal{B}).
\end{equation*}

From assumption 4b

\begin{align*}
E(d^{H}_{i}E(d^B_{i}|d^{H}_{i})|i \in \mathcal{B}) &= E(d^{H}_{i}E(d_{i}|d^{H}=d^{H}_{i}) - d^{H}_{i}|i \in \mathcal{B})\\
&= E(d^{H}_{i}E(d^{B}_{i}|d^{H}=d^{H}_{i})|i \in \mathcal{B}).
\end{align*}

Substituting in the sample analogues and then applying Theorem 1 gives the results.
\end{proof}

\subsection{Proofs of proposition 5}

Given consistency, we now need to derive the asymptotic distribution of the estimator. 

\begin{proof}
As in proof of prop 1., we have

\begin{align*}
\frac{\hat{\beta}^{\text{OLS}}}{1+\eta} &= \frac{1}{1+\eta} \Big(\beta (1+ \eta) + \frac{\frac{1}{N}\sum_{i} (\sum_{j}h_{ij}x_{j})\epsilon_{i}}{\frac{1}{N}\sum_{i} (\sum_{j}h_{ij}x_{j})^{2}}\Big),\\
&= \beta + \frac{1}{1+\eta} \Big(\frac{\frac{1}{N}\sum_{i} (\sum_{j}h_{ij}x_{j})\epsilon_{i}}{\frac{1}{N}\sum_{i} (\sum_{j}h_{ij}x_{j})^{2}}\Big).
\end{align*}

Define the matrices

\begin{align*}
M_{XX} &= \plim \frac{1}{N}\sum_{i} (\sum_{j}h_{ij}x_{j})^{2}\\
M_{X\Omega X} &= \plim \frac{1}{N}\sum_{i} (\sum_{j}h_{ij}x_{j})(\sum_{j}h_{ij}x_{j})\epsilon^{2}_{i}
\end{align*}

Under the maintained assumptions, we can apply the standard proof of the asymptotic distribution of the OLS estimator from \cite{Cameron2005}. This yields

\begin{equation*}
\sqrt{N} \Big(\frac{\frac{1}{N}\sum_{i} (\sum_{j}h_{ij}x_{j})\epsilon_{i}}{\frac{1}{N}\sum_{i} (\sum_{j}h_{ij}x_{j})^{2}}\Big) \sim N(0, M^{-1}_{XX} M_{X\Omega X} M^{-1}_{XX}).
\end{equation*}

Now, applying the normal product rule, we get

\begin{equation*}
\sqrt{N} \Big(\frac{\hat{\beta}^{\text{OLS}}}{1+\eta} - \beta \Big) \sim N(0, (\frac{1}{1+\eta})^{2} M^{-1}_{XX} M_{X\Omega X} M^{-1}_{XX}).
\end{equation*}

\subsection{Proofs of proposition 6, 11}
To prove proposition 6, note that we can write our estimator as a two-step $M$ estimator \citep{Newey1984}.

\begin{align*}
\begin{pmatrix}
h_{1}(\theta)\\
h_{2}(\theta, \beta)
\end{pmatrix} &=
\begin{pmatrix}
\theta - \frac{1}N{}\sum_{i=1}^{N}\theta_{i}\\
\frac{1}{N}\sum_{i}(\sum_{j}h_{ij}x_{j})(y_{i} - (1+\eta(\theta))\beta(\sum_{j}h_{ij}x_{j}))
\end{pmatrix},\\
&=
\begin{pmatrix}
0\\
0
\end{pmatrix}.
\end{align*}

Define

\begin{align*}
\begin{pmatrix}
K_{11} & K_{12}\\
K_{21} & K_{22}
\end{pmatrix}
&= \plim \frac{1}{N} \sum_{i} E
\begin{pmatrix}
-1& 0\\
-  \frac{\partial \eta(\theta)}{\partial \theta}\beta (\sum_{j}h_{ij}x_{j})^{2}  & -(1+\eta(\theta))(\sum_{j}h_{ij}x_{j})^{2}  
\end{pmatrix}\\
\begin{pmatrix}
S_{11} & S_{12}\\
S_{21} & S_{22}
\end{pmatrix} &=
\plim \frac{1}{N} \sum_{i} E \begin{pmatrix}
h_{1i}h_{1i}'& h_{2i}h_{1i}'\\
h_{2i}h_{1i}' & h_{2i}h_{2i}'
\end{pmatrix}
\end{align*}

Assume that 

\begin{align*}
\frac{1}{\sqrt{N}}\sum_{i}h_{1i}(\eta) &\xrightarrow[]{d} N(0,S_{11}(\eta)),\\
\frac{1}{\sqrt{N}}\sum_{i}h_{2i}(\eta, \beta) &\xrightarrow[]{d} N(0,S_{22}(\eta, \beta)).
\end{align*}

We have just shown the second. Assume the first. Then, applying the results in \cite{Newey1984}, we know that therefore

\begin{equation}
\Omega = \operatorname{Var}(\hat{\beta}) = K_{22}^{-1}(S_{22} + K_{21}K_{11}^{-1}S_{11}K_{11}^{-1}K_{21}' - K_{21} K^{-1}_{11}S_{12}  -S_{21}K^{-1}_{11}K_{21}') K_{22}^{-1},
\end{equation}

and 

\begin{equation*}
\sqrt{N}(\hat{\beta} - \beta) = N(0, \Omega).
\end{equation*}

\end{proof}

Proposition 11 follows by the same logic by noting that our estimates for the copula parameters will also satisfy the assumptions for applying the two-step M estimator under standard regularity conditions - see \cite{Smith2003, Choros2010} and references therein.

\subsection{Proof of proposition 7}

Proposition 7 follows by simply rearranging

\begin{equation*}
\frac{\hat{\beta}^{\text{OLS}}}{1+\eta} > \tau
\end{equation*}

for $\hat{\beta}^{\text{OLS}}$.

Results for the non-linear social network model are presented in a separate section later.

\section{Detailed example with fixed choice design}

\textbf{Example -- fixed choice design.} To fix ideas, consider the case of a binary network $h_{ij}, b_{ij} \in \{0,1\}$ where $x_{j} = 1 \text{ } \forall j$. The logic extends to the more general case without loss of generality.\\

The researcher samples up to $m$ links into each individual. For illustration, let $m=5$  \citep[as for same-sex friends in the Ad Health dataset][]{Harris2009}. If an individual has five or fewer connections, the researcher samples all of their connections. Sampled spillovers equal observed spillovers. If an individual has more than five connections, the researcher does not sample some of their spillovers. So they have some positive unobserved spillovers. As they have the maximum number of sampled links, their spillovers are also higher. Individuals with more than five links have a sampled spillover of $5$, greater than or equal to individuals with five or fewer friends (whose spillovers are in $\{0,1,2,3,4,5\}$). Thus, sampling based on generates positive dependence between observed and unobserved spillovers.\\

Formally, we can derive the expected dependence between observed and unobserved spillovers under a fixed choice design as:

\begin{align*}
E \Big(\frac{1}{N}\sum_{i} (\sum_{j}h_{ij}x_{j})(\sum_{j}b_{ij}x_{j}) \Big) &= \frac{1}{N}\sum_{i}E \Big((\sum_{j}h_{ij}x_{j})(\sum_{j}b_{ij}x_{j}) \Big) \text{ by linearity of $E()$},\\
&= \frac{1}{N}\sum_{i} (p(d_{i} \leq m)E \Big((\sum_{j}h_{ij}x_{j})(\sum_{j}b_{ij}x_{j})|d_{i} \leq m \Big), \\
&+ (1-p(d_{i} \leq m))E \Big((\sum_{j}h_{ij}x_{j})(\sum_{j}b_{ij}x_{j})|d_{i} > m \Big)),\\
&= \frac{1}{N}\sum_{i} (1-p(d_{i} \leq m))E \Big((\sum_{j}h_{ij}x_{j})(\sum_{j}b_{ij}x_{j})|d_{i} > m \Big) \text{ as } b_{ij} = 0 \forall j \text{ if }d_{i} \leq m,\\
&= \frac{1}{N}\sum_{i} (1-p(d_{i} \leq m)) m E(d_{i} - m|d_{i} > m) \text{ from the sampling rule,}\\
&= \frac{1}{N}\sum_{i} (1-p(d_{i} \leq m)) m (E(d_{i}|d_{i} > m) - m) > 0.
\end{align*}

Therefore, under this sampling design, estimates are upwards biased ($|\hat{\beta}^{\text{OLS}}| > |\beta|$).\\

\section{Extension to models with covariates}

Here, we derive our results in matrix notation to allow for arbitrary covariates. This allows us to extend the results to general linear regression models, and regression models for panel data. Let

\begin{equation*}
Z = \begin{pmatrix}
Hx \\
W
\end{pmatrix}.
\end{equation*}

Our model in matrix form is

\begin{equation}
y = \begin{pmatrix}
Gx \\
W
\end{pmatrix}'
\begin{pmatrix}
\beta \\
\gamma
\end{pmatrix} + \epsilon.
\label{eq:matrix_dgp}
\end{equation}

The OLS estimator solves

\begin{equation*}
\begin{pmatrix}
\hat{\beta}^{\text{ OLS}}\\
\hat{\gamma}^{\text{ OLS}}
\end{pmatrix} = 
(Z'Z)^{-1}Z'y
\end{equation*}

Solving yields

\begin{align*}
\hat{\gamma}^{\text{OLS}} &= (W'(I-P_{Hx})W)^{-1} W'(I-P_{Hx})y,\\
\hat{\beta}^{\text{OLS}} &= ((Hx)'(I-P_{W})Hx)^{-1} (Hx)'(I-P_{W})y.
\end{align*}

Let $(\overset{\sim}{A})$ denote $(I-P_{W})A$. For readability, write

\begin{equation*}
    \hat{\beta}^{\text{OLS}} = ((\overset{\sim}{Hx})'\overset{\sim}{Hx})^{-1} (\overset{\sim}{Hx})'\overset{\sim}{y}.
\end{equation*}

Substituting (\ref{eq:matrix_dgp}) for $y$ ,

\begin{equation*}
\hat{\beta}^{\text{OLS}} = \beta + ((\overset{\sim}{Hx})'\overset{\sim}{Hx})^{-1} (\overset{\sim}{Hx})'(\overset{\sim}{Bx} \beta + \overset{\sim}{\epsilon}). 
\end{equation*}

Taking expectations

\begin{equation*}
E(\hat{\beta}^{\text{OLS}}) = (I + E((\overset{\sim}{Hx})'\overset{\sim}{Hx})^{-1} (\overset{\sim}{Hx})'\overset{\sim}{Bx})\beta. 
\end{equation*}

Therefore the multiplicative bias is

\begin{equation*}
E((\overset{\sim}{Hx})'\overset{\sim}{Hx})^{-1} (\overset{\sim}{Hx})'\overset{\sim}{Bx}).
\end{equation*}

Equivalents of proposition 1, theorem 1 follow immediately.

Under the same Taylor approximation as in the proof of proposition 4,
\begin{equation*}
E(((\overset{\sim}{Hx})'\overset{\sim}{Hx})^{-1} (\overset{\sim}{Hx})'\overset{\sim}{Bx}) \approx E(((\overset{\sim}{Hx})'\overset{\sim}{Hx})^{-1} )E((\overset{\sim}{Hx})'\overset{\sim}{Bx}),
\end{equation*}
giving the results in section 2.4 for the mean degree of the sampled network projected onto the orthogonal complement of the space of the column space of covariates $W$ and the mean number of missing links after projection onto the orthogonal complement of the space of the column space of covariates $W$.

If we further assume that measurement errors and spillovers are distributed indepedently of covariates throughout

\begin{equation*}
Bx, Gx  \perp\!\!\!\perp W
\end{equation*}

then the results in section 2.4 apply identically. 

In practice, it is important to consider whether this assumption holds or not before bias-correcting the estimator. If it does not, the researcher needs to apply the results using

\begin{align*}
\overset{\sim}{d^{H}} &= \frac{1}{\sum_{i \in \mathcal{B}}1_{i}}\sum_{i \in \mathcal{B}, j}\overset{\sim}h_{ij}\\
\overset{\sim}{d^{B}} &= \frac{1}{\sum_{i \in \mathcal{B}}1_{i}}\sum_{i \in \mathcal{B}, j}\overset{\sim}b_{ij}.\\
\end{align*}

In practice, the researchers could construct these by regressing reported number of links/missing links on covariates amongst all individuals, removing the expectation given the covariates for all individuals, and then taking the mean for individuals with at least some missing links.

We brush over it in the main text for reasons similar to \cite{Battaglia2025} -- considering it directly dilutes the main point of the paper.

In certain cases, including controls can lead to $E(\overset{\sim}{Bx})=0$. In this case, the linear regression estimator is not biased, and correction would be erroneous. Our approach gives a transparent way to see when adding controls will also account for measurement error. An example is a panel data regression with individual fixed effects with constant sampling error by node. In this case 

\begin{align*}
\overset{\sim}{BX}_{it} &= (d^{B}_{it} - \bar{d^{B}_{i}})E(X)\\
&= (d^{B}_{it} - d^{B}_{it} ) E(X)\\
&= 0
\end{align*}

as by construction $d^{B}_{it} = \bar{d^{B}_{i}}$.

\section{Dummy variable estimators}

Again, assume that we can describe the underlying data generating process with (\ref{eq:structural_model}). Instead of estimating the direct spillover effect $\beta$, the researcher wants to estimate the average effect of at least one neighbour being treated on outcomes\footnote{Note that this is a different estimand to the spillover effect $\beta$, though the two are sometimes conflated \citep{Barrot2016}. Different degree distributions of the true underlying network can deliver different $\gamma$ for the same $\beta$.}

\begin{equation*}
\gamma = E(\beta\sum_{j}g_{ij}x_{j}|\sum_{j}g_{ij}x_{j} > 0).
\end{equation*}

For example, the researcher wants to estimate the effect of at least one supplier experiencing a shock on sales  \citep{Barrot2016}. A common estimation strategy is to construct a dummy variable that encodes whether at least one sampled neighbour is treated 

\begin{equation*}
d_{i} = \begin{cases}

1 & \text{ if and only if} \sum_{j}h_{ij}x_{j} \geq 1\\
0 & \text{ else}
\end{cases}
\end{equation*}

and regress on outcomes on this dummy plus an intercept \citep[e.g specifications in][]{Oster2012, Barrot2016} \footnote{We omit controls here without loss of generality.}
By splitting spillovers into observed and unobserved components, we see that this estimator recovers \citep{Angrist2009}

\begin{align*}
\hat{\gamma}^{\text{ OLS}} &=  E(\beta\sum_{j}g_{ij}x_{j}|\sum_{j}h_{ij}x_{j} > 0) -   E(\beta\sum_{j}g_{ij}x_{j}|\sum_{j}h_{ij}x_{j} = 0)\\
& \neq E(\beta\sum_{j}g_{ij}x_{j}|\sum_{j}g_{ij}x_{j} > 0).
\end{align*}

where the second term may be non-zero. 

Again, we can construct an unbiased estimator by rescaling based on the mean number of missing links on the network.

\begin{proposition}
Make assumptions 1,2,3. Consider the estimator

\begin{equation*}
\hat{\gamma} = \frac{\frac{E(d^{H}_{i})+ E(d^{B}_{i})}{p(\sum_{j}g_{ij}x_{j} > 0)}}{\frac{E(d^{H}_{i})} {p(\sum_{j}h_{ij}x_{j} > 0))} + E(d^{B}_{i}| \sum_{j}h_{ij}x_{j} > 0) - E(d^{B}_{i}| \sum_{j}h_{ij}x_{j} = 0)} \hat{\gamma}^{\text{ OLS}}.
\end{equation*}

$\hat{\gamma}$ is an unbiased estimator of $\gamma$.
\end{proposition}

\textbf{Proof}

\begin{proof}
By definition, 

\begin{equation*}
\gamma = \frac{\gamma}{\hat{\gamma}^{\text{ OLS}}}\hat{\gamma}^{\text{ OLS}}.
\end{equation*}

Therefore, $\frac{\gamma}{\hat{\gamma}^{\text{ OLS}}}\hat{\gamma}^{\text{ OLS}}$ is an unbiased estimator of $\gamma$.

Now, we simplify this fraction. Given that outcomes follow (\ref{eq:structural_model}),

\begin{align*}
\gamma &= E(\sum_{j}g_{ij}x_{j} + \epsilon_{i}| \sum_{j}g_{ij}x_{j} > 0) - E(\sum_{j}g_{ij}x_{j} + \epsilon_{i}| \sum_{j}g_{ij}x_{j} = 0)\\
&= E(\sum_{j}g_{ij}x_{j} | \sum_{j}g_{ij}x_{j} > 0) + E( \epsilon_{i}| \sum_{j}g_{ij}x_{j} > 0) - E( \epsilon_{i}| \sum_{j}g_{ij}x_{j} = 0),\\
&= E(\sum_{j}g_{ij}x_{j} | \sum_{j}g_{ij}x_{j} > 0) \text{ by assumption \ref{asm:structural_shocks},}\\
&= \frac{E(\sum_{j}g_{ij}x_{j})}{p(\sum_{j}g_{ij}x_{j} > 0)}\\
&= \frac{E(x)E(\sum_{j}h_{ij} + \sum_{j}b_{ij})}{p(\sum_{j}g_{ij}x_{j} > 0)} \text{ by assumption \ref{asm:distribution_treatment}}\\
&= \frac{E(x)E(d^{H}_{i}+ d^{B}_{i}}{p(\sum_{j}g_{ij}x_{j} > 0)}\\
&= \beta E(x) \frac{(E(d^{H}_{i})+ E(d^{B}_{i}))}{p(\sum_{j}g_{ij}x_{j} > 0)}.
\end{align*}

Similarly

\begin{align*}
\hat{\gamma}^{\text{ OLS}} &= E(\beta \sum_{j}g_{ij}x_{j} + \epsilon_{i}| \sum_{j}h_{ij}x_{j} > 0) - E(\beta \sum_{j}g_{ij}x_{j} + \epsilon_{i}| \sum_{j}h_{ij}x_{j} = 0)\\
&= E(\beta \sum_{j}g_{ij}x_{j}| \sum_{j}h_{ij}x_{j} > 0) - E(\beta \sum_{j}g_{ij}x_{j}| \sum_{j}h_{ij}x_{j} = 0) + E(\epsilon_{i}| \sum_{j}h_{ij}x_{j} > 0) - E(\epsilon_{i}|\sum_{j}h_{ij}x_{j} = 0),\\
&= \beta (E(\sum_{j}h_{ij}x_{j} + \sum_{j}b_{ij}x_{j}| \sum_{j}h_{ij}x_{j} > 0) - E(h_{ij}x_{j} + \sum_{j}b_{ij}x_{j}| \sum_{j}h_{ij}x_{j} = 0)),\\
&= \beta (E(\sum_{j}h_{ij}x_{j}| \sum_{j}h_{ij}x_{j} > 0) + E(\sum_{j}b_{ij}x_{j}| \sum_{j}h_{ij}x_{j} > 0) - E(\sum_{j}b_{ij}x_{j}| \sum_{j}h_{ij}x_{j} = 0))\\
&= \beta E(x)(E(\sum_{j}h_{ij}| \sum_{j}h_{ij}x_{j} > 0) + E(\sum_{j}b_{ij}| \sum_{j}h_{ij}x_{j} > 0) - E(\sum_{j}b_{ij}| \sum_{j}h_{ij}x_{j} = 0)) \text{ by assumption \ref{asm:distribution_treatment}},\\
&= \beta E(x) \Big(\frac{E(\sum_{j}h_{ij})} {p(\sum_{j}h_{ij}x_{j} > 0))} + E(\sum_{j}b_{ij}| \sum_{j}h_{ij}x_{j} > 0) - E(\sum_{j}b_{ij}| \sum_{j}h_{ij}x_{j} = 0) \Big) \text{ by assumption \ref{asm:distribution_treatment}},\\
&= \beta E(x) \Big(\frac{E(d^{H}_{i})} {p(\sum_{j}h_{ij}x_{j} > 0))} + E(d^{B}_{i}| \sum_{j}h_{ij}x_{j} > 0) - E(d^{B}_{i}| \sum_{j}h_{ij}x_{j} = 0) \Big)
\end{align*}

Therefore

\begin{align*}
\gamma &= \frac{\gamma}{\hat{\gamma}^{\text{ OLS}}}\hat{\gamma}^{\text{ OLS}},\\
&= \frac{\beta E(x) \frac{(E(d^{H}_{i})+ E(d^{B}_{i}))}{p(\sum_{j}g_{ij}x_{j} > 0)}}{\beta E(x) \Big(\frac{E(d^{H}_{i})} {p(\sum_{j}h_{ij}x_{j} > 0))} + E(d^{B}_{i}| \sum_{j}h_{ij}x_{j} > 0) - E(d^{B}_{i}| \sum_{j}h_{ij}x_{j} = 0) \Big)}  \hat{\gamma}^{\text{ OLS}}\\
&= \frac{\frac{E(d^{H}_{i})+ E(d^{B}_{i})}{p(\sum_{j}g_{ij}x_{j} > 0)}}{\frac{E(d^{H}_{i})} {p(\sum_{j}h_{ij}x_{j} > 0))} + E(d^{B}_{i}| \sum_{j}h_{ij}x_{j} > 0) - E(d^{B}_{i}| \sum_{j}h_{ij}x_{j} = 0)} \hat{\gamma}^{\text{ OLS}}.
\end{align*}

\end{proof}

Sample analogues for $E(d^{H}_{i}), \frac{E(d^{H}_{i})} {p(\sum_{j}h_{ij}x_{j} > 0))} $ are directly computable from observed $H, x$. The other missing terms -- the expected number of unobserved links, and difference in the the expected number of unobserved links between individuals with at least one sampled treated neighbour and individuals with no sampled treated neighbours -- are again aggregate network statistics. The researchers can construct sample analogues for the other terms. They can do this by asking each individual how many connections they have in a survey, disclosed by data providers without violating privacy, or approximated from detailed sampling of similar datasets.

\section{Equivalence to control function approach}

Writing out our data-generating process again, we have that

\begin{align*}
y_{i} &= \beta \sum_{j}g_{ij}x_{j} + \epsilon_{i}\\
&= \beta \sum_{j}h_{ij}x_{j} + \xi_{i}
\end{align*}

where 

\begin{equation*}
\xi_{i} = \sum_{j}b_{ij}x_{j}\beta + \epsilon_{i}. 
\end{equation*}

A model for the error under assumption 1 is

\begin{equation*}
E(\xi_{i}) =  d^{B}_{i}E(x_{j}).
\end{equation*}

The resulting regression model would be

\begin{equation*}
y_{i} = \beta \sum_{j}g_{ij}x_{j} +  \gamma d^{B}_{i}E(x_{j}).
\end{equation*}

which gives the same regression estimator as in the main text. Of course, this requires knowing which individuals have at least some incorrectly sampled links. Thus, it is only implementable for a subset of the sampling schemes used to study economic networks (e.g fixed choice designs, but not assuming that all individuals in the same group are connected).

\section{Results for non-linear social network models}

Make the standard assumptions  \citep{Kelejian1998, Bramoulle2009, Blume2015}. 
\paragraph{Assumption A2}
    \label{asm:SAR_assumptions}
Assume that
\begin{enumerate}
\item $(y, G, B, x)$ are independently but not identically distributed over $i$,
    \item $E(\epsilon| G, x) = 0$
    \item $\epsilon$ are independent and not identically distributed over $i$ such that for some $\delta > 0$ $E(|u^{2}_{i}|^{1+\delta}) < \infty$ with conditional variance matrix 
    \begin{equation*}
E(\epsilon \epsilon'|(G-B)x) = \Omega
    \end{equation*}
    which is diagonal.
\item \begin{align*}
\text{plim }N^{-1}z'P_{J*}z &= Q_{ZZ}\\
\text{plim }N^{-1}z'P_{J*}z_{B} &= Q_{ZB}\\ 
\text{plim }N^{-1}z'P_{J*} &= Q_{HJ}
\end{align*}

which are each finite nonsingular.
%\item The sequence of networks $\{G^{*}, B\}_{N}$ are uniformly bounded simple networks.
\item $|\lambda| < \frac{1}{||H||}, \frac{1}{||G||}$ for any matrix norm $||.||$.
\end{enumerate}

The estimator for non-linear social network models given in Section 3 is consistent and asymptotically normal.

\begin{theorem}
Consider the debiased estimator $\hat{\theta}$, and make assumption \ref{asm:SAR_assumptions}. Then $\text{plim }\hat{\theta} = \theta$ and
\begin{equation*}
\frac{1}{\sqrt{N}}(\hat{\theta} - \theta) \xrightarrow[]{d} \mathcal{N}(0, N(0, \sigma^2 (I + Q_{ZZ}^{-1} Q_{ZB})^{-1} Q^{-1}_{ZZ} Q_{HJ} ((I + Q_{ZZ}^{-1} Q_{ZB})^{-1} Q^{-1}_{ZZ})'),
\end{equation*}
where 
\begin{align*}
\text{plim }N^{-1}Z'P_{J*}Z &= Q_{ZZ}\\
\text{plim }N^{-1}Z'P_{J*}Z_{B} &= Q_{ZB}\\ 
\text{plim }N^{-1}Z'P_{J*} &= Q_{HJ}
\end{align*}

\end{theorem}

\begin{proof}

Let $z^{*} = \begin{pmatrix} Gy, x \end{pmatrix}$, $z = \begin{pmatrix} Hy, x \end{pmatrix}$. Call $z_{B} = z^{*} - z = \begin{pmatrix} By, 0 \end{pmatrix}$. Finally, denote the projection matrix onto the space spanned by our instruments $P_{J*} = J^{*} (J^{*'}J^{*})^{-1}J^{*'}$ .

Our two-stage least squares estimates with our unbiased instruments $J^{*}$ are 
\begin{align*}
\hat{\theta}^{2sls} &= ((P_{J*}z)'P_{J*}z)^{-1}(P_{J*}z)'y,\\
&= ((P_{J*}z)'P_{J*}z)^{-1}(P_{J*}z)'(z^{*}\theta + \epsilon)\\
&= (z'P_{J*}z)^{-1}(P_{J*}z)'(z\theta + z_{B}\theta+ \epsilon)\\
&= \theta + ((z'P_{J*}z)^{-1}(P_{J*}z)'z_{B}\theta + (z'P_{J*}z)^{-1}(P_{J*}z)'\epsilon.
\end{align*}

Therefore, 

\begin{equation*}
\hat{\theta} = (I + (z'P_{J*}z)^{-1}(z'P_{J*}z_{B}))^{-1}\hat{\theta}^{2sls} = \theta + (I + (z'P_{J*}z)^{-1}z'P_{J*}z_{B})^{-1}(z'P_{J*}z)^{-1}(P_{J*}z)'\epsilon.
\end{equation*}

Note that 

\begin{equation*}
z'P_{J*}z_{B} = \begin{pmatrix}
0 & (Hy)'P_{J*}By\\
0 & x'P_{J*}By
\end{pmatrix}.
\end{equation*}

First, we show the consistency of this estimator. As per assumption \ref{asm:SAR_assumptions}

\begin{align*}
\text{plim }N^{-1}z'P_{J*}z &= Q_{ZZ}\\
\text{plim }N^{-1}z'P_{J*}z_{B} &= Q_{ZB}\\ 
\text{plim }N^{-1}z'P_{J*} &= Q_{HJ}
\end{align*}

which are each finite nonsingular.

Therefore

\begin{align*}
\text{plim } \hat{\theta} &= \text{plim }(\theta + (I + (N^{-1}z'P_{J*}z)^{-1}N^{-1}z'P_{J*}z_{B})^{-1}(N^{-1}z'P_{J*}z)^{-1}(N^{-1}P_{J*}z)'\epsilon)\\
&=\theta + (I + Q_{ZZ}^{-1} Q_{ZB})^{-1} Q^{-1}_{ZZ} \text{plim} N^{-1}z'P_{J*}\epsilon\text{  by Slutsky's lemma}\\
\end{align*}

Finally, we need to characterise the properties of 

\begin{equation*}
\text{plim} N^{-1}z'P_{J*}\epsilon.
\end{equation*}

\begin{equation*}
N^{-1}z'P_{J*} = \begin{pmatrix}
N^{-1}(P_{J*}Gy)'\epsilon\\
N^{-1} (P_{J*}x)' \epsilon
\end{pmatrix}.
\end{equation*}

We can characterise the behaviour of the second row using a standard weak law of large numbers. But, the vector $Gy$ involves a sum of random variables $y$. So, here, we need to apply a law of large numbers for triangular arrays. 
From assumption \ref{asm:SAR_assumptions}, it follows that the array $G_{1,1} y_{1}, G_{1,2} y_{2}, ... $ is a triangular array \citep{Kelejian1998}. So, the term $GY)'\epsilon$ is the sum of 

\begin{equation*}
(G_{1,1} y_{1}, G_{1,2} y_{2}, ... ) \epsilon_{1} + (G_{2,1} y_{1}, G_{2,2} y_{2}, ... ) \epsilon_{2} + ...
\end{equation*}

which is itself a triangular array. Call this triangular array $W$. Assume that $\text{sup}_{N}E_{N}(W^{2})< \infty$ for all $N$. Then we can apply a weak law of large numbers for triangular arrays to $W$ to say that

\begin{equation*}
\text{plim }N^{-1}(P_{J*}Gy)'\epsilon = E((P_{J*}Gy)'\epsilon)_{i}) = 0.
\end{equation*}

Therefore our estimator is both unbiased and consistent.

Next, we need to characterise the asymptotic distribution of the estimator.

\begin{equation*}
\sqrt{N}(\hat{\theta} - \theta) = (I + (N^{-1}z'P_{J*}z)^{-1}N^{-1}z'P_{J*}z_{B})^{-1}(N^{-1}z'P_{J*}z)^{-1}(\frac{1}{\sqrt{N}}P_{J*}z)'\epsilon)\\
\end{equation*}

Again, applying Slutsky's lemma, all terms on the right hand side except 

\begin{equation*}
\frac{1}{\sqrt{N}}(P_{J*}z)'\epsilon
\end{equation*}

will converge to finite limits. To characterise the distribution of this term, we need to apply a law of large numbers for triangular arrays. We use the central limit theorem for triangular arrays from \citep{Kelejian1998}.

\begin{theorem}[CLT for triangular arrays]
    Let $\epsilon$, $P_{J*}Hy$ be triangular arrays of identically distributed random variables with finite second moments. Denote $\text{Var}(\epsilon) = \sigma^{2}$. Assume that $\text{plim } N^{-1}(P_{J*}Hy)'P_{J*}Hy = Q_{HJ}$ is finite and nonsingular. Then 
    \begin{equation*}
    \frac{1}{\sqrt{N}}(P_{J*}z)'\epsilon \xrightarrow{d} N(0, \sigma^2 Q_{HJ}).
    \end{equation*}
\end{theorem}

Applying this result, we have that

\begin{equation*}
\frac{1}{\sqrt{N}}(P_{J*}z)'\epsilon \xrightarrow{d} N(0, \sigma^2 Q_{HJ}).
\end{equation*}

Therefore, by Slutky's lemma

\begin{equation*}
\sqrt{N}(\hat{\theta} - \theta) \xrightarrow{d} N(0, \sigma^2 (I + Q_{ZZ}^{-1} Q_{ZB})^{-1} Q^{-1}_{ZZ} Q_{HJ} ((I + Q_{ZZ}^{-1} Q_{ZB})^{-1} Q^{-1}_{ZZ})'). 
\end{equation*}

\end{proof}

%%%%%%%%%%%%%%%%%%%%%%%%%%%%%%%%%%%%%%%%%%%%%%
\section{Additional simulations}
%%%%%%%%%%%%%%%%%%%%%%%%%%%%%%%%%%%%%%%%%%%%%%
\subsection{Real-data simulation}

We further test the performance of our estimator on a real network - the co-author network of economists from \cite{Ductor2014}. This is the complete network of co-authorships between economists on papers published in journals in the EconLit database. As in \cite{Ductor2014}, we use co-authorships over a three-year window -- here 1996-1998 -- to account for lags in publications. This gives us across $44,776$ economists and $57,407$ links between them. Note that the network is very sparse. The mean degree is $1.28$. The $95$th percentile of the degree distribution is $4$ collaborations.

We simulate the effect of a treatment across this network as above. In each simulation, each economist draws a binary treatment $x_{i} \sim \text{Bernoulli}(0.3)$. Outcomes are drawn from (\ref{eq:structural_model})
with $\beta = 0.8$, $\epsilon_{i} \sim N(0,1)$.

We sample the network using a fixed choice design with thresholds $k \in \{1,2,3,4,5,6\}$. Next, we sample based on groups. We then construct spillover estimates using the sampled network, and using our debiased estimator under assumption 4.a.

\begin{figure}[!htbp]
 \begin{minipage}{\textwidth}
  \begin{minipage}[b]{0.49\textwidth}
    \centering
\input{coauthor_table_1}
\label{tab:fig1} 
\caption{Mean spillover estimates using fixed choice design, by threshold}

   \end{minipage}
    \begin{minipage}[b]{0.49\textwidth}
\includegraphics[width=0.95\textwidth]{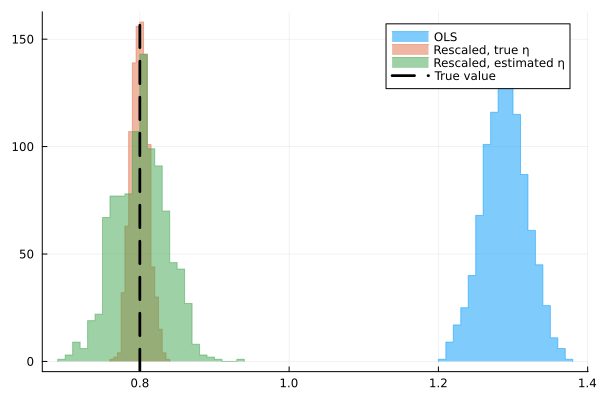}
\caption{Distribution of spillover estimates using fixed choice design with threshold of $3$}
    \end{minipage}
\end{minipage}

\end{figure}

As in our simulated networks, we see that linear regression of outcomes on sampled spillovers leads to biased estimates. The bias is relatively small because the true network is so sparse. With a threshold of $3$, $90\%$ of individuals maintain all of their true links. Our error corrected estimate performs still perform very well. 

\subsection{Simulations for nonlinear social network models}

We test the performance of our estimator in Section \ref{sec:nonlinear_models} in finite sample. As in the experiments in the main text, we simulate $N=1000$ individuals who draw a true degree $d_{i}\sim U(0,10)$ and are then connected with others uniformly at random from the population.

Our data generating process is

\begin{equation*}
y = \lambda Gy + x\beta + \epsilon
\end{equation*}

with $\lambda = 0.3, \beta = 0.8$. In all cases, $\epsilon \sim N(0,1)$. We run $1000$ simulations per estimator, starting each set with the same random seed. Bias corrected estimators are constructed using the mean missing degree $d^{B}$ under Assumption 5 for cases 1 and 2 in Section 5.

\begin{figure}
\caption{Spillover estimates from nonlinear social network models}
 \label{fig:mc_add_health}
 \centering
 \hspace{.05 \linewidth}\subfloat[\label{fig1:a}][Case 1]
 {\includegraphics[width=.45\linewidth]{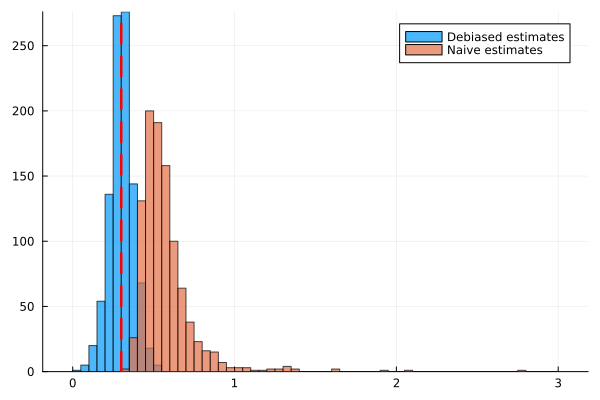}}%\hfill
 \subfloat[\label{fig1:b}][Case 2]
 {\includegraphics[width=.45\linewidth]{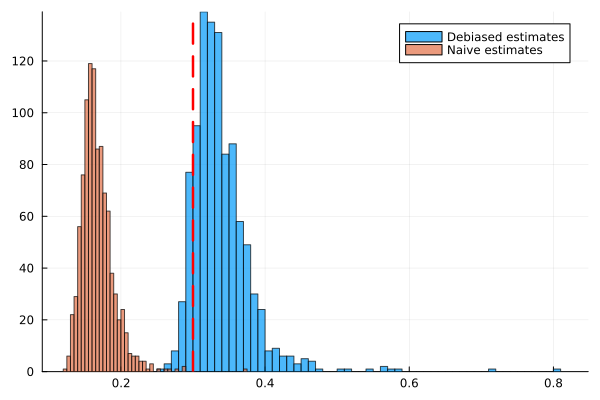}} \hspace{.05 \linewidth}
    \footnotesize
    \textbf{Notes:}  
    Red line denotes true parameter value. Sampled network in left panel generated by sampling $5$ links per agent uniformly at random from their true links. Sampled network in right panel generated by sampling $10-d_{i}$ additional links per agent $i$ uniformly at random from the population. 

  \end{figure}

As for linear models, we see that naive two-stage least squares estimators using the sampled network are heavily biased. Our bias-corrected estimators recover the true spillover effect well in finite samples. 

\subsection{Copula-based estimator}
We assess the performance of an example of this estimator in section 4 in finite sample. As above, we simulate $N=1000$ individuals who draw a true degree $d_{i}\sim U(0,10)$ and are then connected with others uniformly at random from the population.

Each agent draws continuous treatment from the marginal distribution $X_{i} \sim N(5,1)$. Marginal distributions of treatment and degree are coupled through a bivariate Gumbel copula 
\begin{equation*}
C(F_{X}^{-1}(x), F_{D}^{-1}(d); \theta) = \text{exp}(-((-\ln{F_{X}^{-1}(x)})^{\theta} + (- \ln{F_{D}^{-1}(d)})^{\theta})^{\frac{1}{\theta}})
\end{equation*}

where $\theta \in [1, \infty]$ controls the degree of dependence between treatment and degree. We set $\theta = 10$. The left panel of figure \ref{fig:mc_copula} plots an example joint distribution. Higher treatment nodes have higher degree. Researchers sample networks using a fixed choice design sampling $m=5$ links per node as in the National Longitudinal Survey of Adolescent Health Data Set. Then
\begin{equation*}
\sum_{j}E(b_{ij}(x_{i})|x_{j}) x_{j} = \sum_{j} (E(g_{ij}^{*}|x_{i}) - m)\bar{x}.
\end{equation*}

\begin{figure}
\caption{Spillover estimates when degree depends on treatment}
 \label{fig:mc_copula}
\hspace{.05 \linewidth}\subfloat[\label{fig1:a}][Joint distribution of degree and treatment]
 {\includegraphics[width=.45\linewidth]{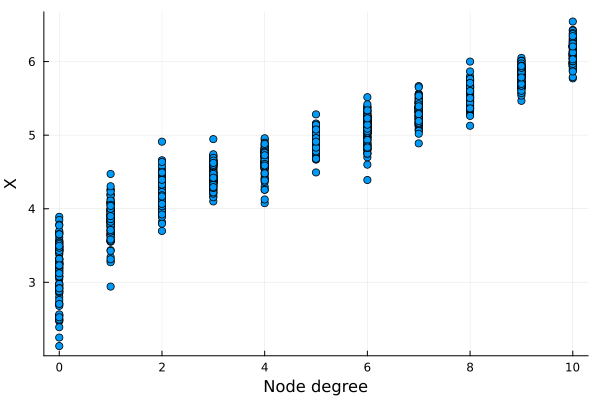}}%\hfill
 \subfloat[\label{fig1:b}][Spillover estimates]
 {\includegraphics[width=.45\linewidth]{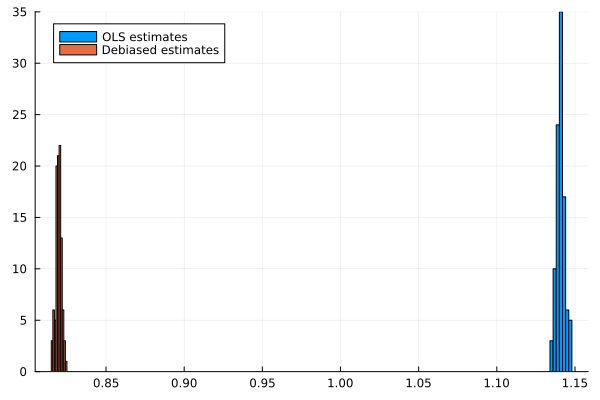}} \hspace{.05 \linewidth}
    \footnotesize
    \textbf{Notes:}  
    Red line denotes true parameter value of $0.8$. Data is simulated from a linear model on the true network with $N = 1000$. Treatment drawn from marginal $N(5,1)$, and degree distributed $U(0,10)$, coupled by a Gumbel copula with $\theta = 10$. Sampled network generated by sampling $5$ links per agent uniformly at random from their true links, or all if degree is less than $5$. 

  \end{figure}

We estimate spillovers using the two-step estimator we describe above. In the first step, we estimate the dependence between treatment and degree by fitting a Gumbel copula by maximum likelihood using only the observations where we correctly sample the network. In the second stage, we then construct a spillover estimate $\hat{\beta}$, constructing $BX$ by sampling from the copula.

Our two-step estimator performs well even though the ordinary least-squares estimator does not. The mean debiased estimate of $0.813$ is close to the true spillover value. 

\FloatBarrier
%%%%%%%%%%%%%%%%%%%%%%%%%%%%%%%%%%%%%%%%%%%%%%%%%%%%%%%%%
\section{Peer effects from classrooms}
\label{sec:peer_effects}
%%%%%%%%%%%%%%%%%%%%%%%%%%%%%%%%%%%%%%%%%%%%%%%%%%%%%%%%%

\cite{Carrell2013} estimate the effect of the share of (randomly assigned) high and low ability peers on student GPA at the United States Air Force Academy assuming that all individuals within a peer group (squadron) influence each other equally. 

Specifically, each student $i$ is placed within one squadron $S_{i}$ with $30$ other individuals. Denote whether a student has high, middle, or low predicted GPA with the dummies $\{D^{H}, D^{M}, D^{L}\}$, whether they have a high SAT-Verbal score with the dummy $x^{H}$, and whether they have a low SAT-Verbal score with the dummy $x^{L}$.

The sampled network of peers $G$ is a binary network such that $G_{ij} = 1$ if and only if $i$ and $j$ are in the same squadron. Treatments are the high-ability and low-ability peers in the same squadron $\mathbbm{1}(S_{i} = S_{j})x^{H}_{j}$, $\mathbbm{1}(S_{i} = S_{j})x^{L}_{j}$.  Students are assigned randomly to squadrons. Therefore sampled spillovers from high-low SAT-Verbal peers are
\begin{equation*}
S^{k}_{i} = \frac{1}{|\mathcal{S}_{i}| - 1}\sum_{j}G_{ij} \mathbbm{1}_{S_{i} = S_{j}}x^{k}_{j}
\end{equation*}

for $k \in \{H, L\}$ where normalising by $\frac{1}{|\mathcal{S}_{i}| - 1}$ give the share of that type of peer in the squadron.

\cite{Carrell2013} then estimate spillover coefficients for each predicted-GPA group using the reduced-form regression
\begin{equation*}
GPA_{i} = W \gamma + \sum_{l}\sum_{k} D_{l}S^{k} \beta_{kl} + \epsilon_{i}.
\end{equation*}

They use the results to run a treatment where they assign new students to squadrons to maximise the GPA of students with the lowest GPA. Using estimated $\hat{\beta}^{OLS}_{HL}, \hat{\beta}^{OLS}_{LL} = 0.464, 0.065$  predicts a positive average treatment effect 
\begin{align*}
\Delta S^{H} \times \beta^{LH} + \Delta S^{L} \times \beta^{LH} &= 0.0464 + 0.006600\\
&= 0.053 > 0
\end{align*}

on the students with the lowest GPA, where $\Delta S^{H} = 0.1, \Delta S^{L} = 0.1015$. Surprisingly, they instead find a negative treatment effect.

One reason reassignment might have less positive effects than expected is that different types interact with different intensities. For example, students may interact less intensely with students with low SAT verbal scores than implied by their shares in the squadron, and more intensely with students with high SAT verbal scores than their shares in the squadron. 

\cite{Jackson2022} survey the network of most important study partnerships between Caltech students, and compute shares of study partners across the GPA distribution. There are $36.28\%$ more study partnerships between students above and below the median on the GPA distribution than implied by their shares in the population. To investigate how sampling of the initial network might affect the \cite{Carrell2013} results, take this as an initial prediction for missing interactions between low predicted GPA and high SAT verbal students.\footnote{Note that \cite{Carrell2013} define high, medium, and low in terms of thirds of the distribution. So, these are not directly comparable. Instead, it can be viewed as a best approximation to the level of sampling bias.} Then, taking values from Tables 1 and 2 in \cite{Carrell2013} gives an estimate of $\beta^{LH}$ of

\begin{align*}
\hat{\beta} &= \frac{0.464}{1+\frac{\bar{S^{H}}^{2} \times 0.3628}{\operatorname{Var}(S^{H})}}\\
&= 0.07709.
\end{align*}

Then, the predicted treatment effect would be 
\begin{equation*}
0.007709 + 0.006600 = 0.01431,
\end{equation*}

a null effect given the forecast standard errors reported in Table 4.

In the paper, they find a negative treatment effect. So, sampling bias cannot entirely rationalise the results. But, it goes a way to explaining how the relatively small amount of endogeneous network adjustment reported in response to treatment could explain the negative result.

\subsection{Calculations from Caltech cohort study}

From \cite{Jackson2022}, there are an average of $3.5$ study partners for male students, and $3.3$ for female students. $65.23\%$ of the cohort are male, and $34.77\%$ are female. So, the average number of study partners is

\begin{equation*}
3.5 \times 0.6523 + 3.3 \times 0.3477 = 3.43.
\end{equation*}

$893$ students answered the survey in $2014$. Therefore 

\begin{equation*}
893 \times 3.43 = 3063
\end{equation*}

study links exist between students. The study network is a simple network. Therefore, there are ${893 \choose 2} = 398278$ possible links. The number of links present per 1000 possible links is therefore

\begin{equation*}
\frac{3063}{ 398278} \times 1000 = 7.69.
\end{equation*}

In Table 4, \cite{Jackson2022} report that there are $2.79$ fewer links per 1000 potential links between pairs of students that both have above/below median GPA than pairs of students with GPA on opposite sides of the median. As there are $7.69$ links on average, if links were drawn uniformly at random across students there would be

\begin{equation*}
\frac{7.69}{2} = 3.845
\end{equation*}

links within and across the GPA categories. The results imply that instead there are

\begin{equation*}
3.845 - \frac{2.79}{2} = 2.45
\end{equation*}

links within the GPA categories, and

\begin{equation*}
3.845 + \frac{2.79}{2} = 5.24
\end{equation*}

links across the GPA categories. This is

\begin{equation*}
\frac{5.24 - 3.845}{3.845} \times 100 = 36.28 \%
\end{equation*}

more than implied by the shares in the population.

\FloatBarrier

\end{document}

%% file: mc_table_1.tex
\begin{tabular}{cccc}
\toprule
Number sampled &OLS& $\eta$ & $\hat{\eta}$\tabularnewline
\midrule
3&1.67&0.800&0.800\tabularnewline
4&1.46&0.800&0.800\tabularnewline
5&1.28&0.800&0.800\tabularnewline
6&1.14&0.800&0.800\tabularnewline
7&1.08&0.800&0.800\tabularnewline
8&1.00&0.800&0.800\tabularnewline
9&0.950&0.800&0.800\tabularnewline
10&0.900&0.800&0.800\tabularnewline
\bottomrule
\end{tabular}

%% file: mc_table_2.tex
\begin{tabular}{cccc}
\toprule
k &OLS& $\eta$ & $\hat{\eta}$\tabularnewline
\midrule
1&0.700&0.800&0.770\tabularnewline
2&0.660&0.800&0.780\tabularnewline
3&0.630&0.800&0.780\tabularnewline
4&0.590&0.800&0.770\tabularnewline
5&0.550&0.800&0.770\tabularnewline
\bottomrule
\end{tabular}

%% file: mc_table_3.tex
\begin{tabular}{cccc}
\toprule
Threshold &OLS& $\eta$ & $\hat{\eta}$\tabularnewline
\midrule
0.200&1.90&0.800&0.780\tabularnewline
0.175&1.92&0.810&0.750\tabularnewline
0.150&1.78&0.800&0.730\tabularnewline
0.120&1.63&0.790&0.710\tabularnewline
0.100&1.54&0.810&0.710\tabularnewline
0.075&1.46&0.820&0.710\tabularnewline
0.050&1.36&0.800&0.690\tabularnewline
0.025&1.31&0.810&0.690\tabularnewline
\bottomrule
\end{tabular}

%% file: mc_table_4.tex
\begin{tabular}{cccc}
\toprule
Number sampled &OLS& $\eta$ & $\hat{\eta}$\tabularnewline
\midrule
3&0.990&0.800&0.800\tabularnewline
4&0.960&0.800&0.800\tabularnewline
5&0.920&0.800&0.800\tabularnewline
6&0.880&0.800&0.800\tabularnewline
7&0.880&0.800&0.800\tabularnewline
8&0.860&0.800&0.800\tabularnewline
9&0.850&0.800&0.800\tabularnewline
10&0.830&0.800&0.800\tabularnewline
\bottomrule
\end{tabular}

%% file: mc_table_5.tex
\begin{tabular}{cccc}
\toprule
k &OLS& $\eta$ & $\hat{\eta}$\tabularnewline
\midrule
1&0.650&0.800&0.780\tabularnewline
2&0.560&0.800&0.770\tabularnewline
3&0.470&0.800&0.760\tabularnewline
4&0.380&0.800&0.730\tabularnewline
5&0.290&0.800&0.710\tabularnewline
\bottomrule
\end{tabular}

%% file: weather_shock_tab.tex
\begin{tabular}{@{\extracolsep{5pt}}lccc} 
\hline \\[-1.8ex]
Disaster & Date & Damages (Billions, 2024 Dollars)& States declaring states of emergency\\
\hline \\[-1.8ex] 
Southern Tornado Outbreak & January 20-22 & 1.4 & GA, MS\\
Missouri and Arkansas Flooding & April 25--May 7 & 2.2 & AR, MO \\
North Central Severe Weather and Tornadoes & May 15-18 & 1.2 & OK\\
Hurricane Harvey & August 25-31 & 160.0 & TX, LA \\
Hurricane Irma & September 6-12 & 64.0 & FL, GA, SC \\
Hurricane Maria & September 19--21 & 115.2 & GA \\
\hline
\end{tabular} 

%% file: reg_tab.tex
\begin{tabular}{@{\extracolsep{5pt}}lcccc}
\toprule 
\hline 
\vspace{-0.35cm} \\
& & $\Delta \ln{\text{Sales}}$  & & \\
\vspace{-0.35cm} \\
\hline
Estimator & OLS  & OLS & Rescaled (Factset) & Rescaled (Herskovic et al.) \\
\hline 
\vspace{-0.35cm} \\
Suppliers shocked  & $-0.00675$ & $-0.0248$& $-0.0140$ & $-0.0132$  \\
& $(0.00303)$ & $(0.0114)$ &$(0.01)$ & $(0.01)$  \\
Shocked  & $0.0460$& $0.0650$ & $0.0650$ & $0.0650$ \\
& $( 0.0464)$ & $(0.0608)$ &  $(0.0608)$&  $(0.0608)$  \\
Size & Yes & Yes & Yes & Yes  \\
Industry Fixed Effects & No & Yes & Yes & Yes  \\
State Fixed Effects & No & Yes & Yes & Yes  \\
\vspace{-0.35cm} \\
\hline
\vspace{-0.35cm} \\
Obs & $1711$ &$1243$ & $1243$ & $1243$ \\
$R^{2}$ & $0.001$ & $0.103$ & $0.103$& $0.103$\\
\hline
\bottomrule
\end{tabular}

%% file: coauthor_table_1.tex
\begin{tabular}{cccc}
\toprule
Number sampled &OLS& $\eta$ & $\hat{\eta}$\tabularnewline
\midrule
1&1.02&0.799&0.798\tabularnewline
2&0.902&0.800&0.799\tabularnewline
3&0.871&0.800&0.800\tabularnewline
4&0.847&0.799&0.799\tabularnewline
5&0.834&0.800&0.800\tabularnewline
6&0.823&0.799&0.799\tabularnewline
\bottomrule
\end{tabular}